\documentclass[sigconf]{acmart}

\usepackage[ruled, noend, noline, linesnumbered]{algorithm2e}
\usepackage[noabbrev, capitalise]{cleveref}
\usepackage{paralist}
\usepackage{xcolor}
\usepackage{multirow}

\usepackage{subcaption}

\usepackage{arydshln}

\newcommand{\meach}{\text{\textnormal{\textbf{each} }}}

\newcommand{\ita}[1]{{\it #1}}
\newcommand{\ti}[1]{{\tt #1}}
\newcommand{\mb}[1]{$\mathbf{\textit{#1}}$\,}

\newcommand{\NN}{\mathbb{N}}

\newcommand{\quark}[1]{{\mathcal K({#1})}}
\newcommand{\quarkRole}[2]{{\mathcal K_#1({#2})}}

\newcommand{\MC}{\textsc{MC}\xspace}

\newcommand{\MCsingle}{\textsc{MC-Single}\xspace}
\newcommand{\MCrecbi}{\textsc{MC-Rec-Bisection}\xspace}
\newcommand{\MCkmeans}{\textsc{MC-k-means}\xspace}
\newcommand{\TY}{\textsc{TY}\xspace}

\newcommand\mycommfont[1]{\textcolor{blue}{#1}}

\newtheorem{definition}{Definition}

\newtheorem{theorem}{Theorem}

\AtBeginDocument{%
  \providecommand\BibTeX{{%
    \normalfont B\kern-0.5em{\scshape i\kern-0.25em b}\kern-0.8em\TeX}}}

\setcopyright{iw3c2w3} 

\copyrightyear{2021}
\acmYear{2021} 
\acmConference[WWW '21]{Proceedings of the Web Conference 2021}{April 19--23, 2021}{Ljubljana, Slovenia} 
\acmBooktitle{Proceedings of the Web Conference 2021 (WWW '21), April 19--23, 2021, Ljubljana, Slovenia}
\acmPrice{}
\acmDOI{10.1145/3442381.3450055}
\acmISBN{978-1-4503-8312-7/21/04}

\begin{document}


\title{Motif-driven Dense Subgraph Discovery\\ in Directed and Labeled Networks}


\author{Ahmet Erdem Sar{\i}y\"{u}ce}
\affiliation{ 
  \institution{University at Buffalo}
  \city{Buffalo}
  \state{NY}
  \country{USA}
}
\email{erdem@buffalo.edu}


\begin{abstract}
Dense regions in networks are an indicator of interesting and unusual information.
However, most existing methods only consider simple, undirected, unweighted networks.
Complex networks in the real-world often have rich information though: edges are asymmetrical and nodes/edges have categorical and numerical attributes.
Finding dense subgraphs in such networks  in accordance with this rich information is an important problem with many applications.
Furthermore, most existing algorithms ignore the higher-order relationships (i.e., motifs) among the nodes.
Motifs are shown to be helpful for dense subgraph discovery but their wide spectrum in heterogeneous networks makes it challenging to utilize them effectively.
In this work, we propose quark decomposition framework to locate dense subgraphs that are rich with a given motif.
We focus on networks with directed edges and categorical attributes on nodes/edges.
For a given motif, our framework builds subgraphs, called quarks, in varying quality and with hierarchical relations.
Our framework is versatile, efficient, and extendible.
We discuss the limitations and practical instantiations of our framework as well as the role confusion problem that needs to be considered in directed networks.
We give an extensive evaluation of our framework in directed, signed-directed, and node-labeled networks.
We consider various motifs and evaluate the quark decomposition using several real-world networks.
Results show that quark decomposition performs better than the state-of-the-art techniques. Our framework is also practical and scalable to networks with up to 101M edges.
\end{abstract}

%
%
%





\maketitle

\section{Introduction}\label{sec:intro}

Dense regions in networks contain unusual and interesting information~\cite{Gleich12}.
Dense subgraph discovery is shown to be an effective analysis method in many applications across various domains~\cite{Lee10, Kumar99, Alvarez06, Fratkin06, Du09, Jin09, Gionis13}.
It is often a good and cheaper proxy for graph clustering because cohesive subgraphs in real-world networks exhibit good cuts~\cite{Lee10, Gleich12}.
However, most algorithms to find dense subgraphs are designed for simple, undirected, unweighted graphs.
In reality, a common characteristic of natural and engineered systems from various domains is that the nodes (entities) and edges (relationships) have rich information associated with them; i.e., networks are heterogeneous~\cite{sun2012mining}.
Relationships can be asymmetrical (one-way) and entities/relationships can be associated with categorical and numerical attributes; e.g., length of a road or gender of a person.
Finding the dense subgraphs while actively considering the rich information on nodes/edges has various applications, such as entity resolution~\cite{wang2016efficient} and link prediction~\cite{carranza2020higher}. 

Furthermore, most dense subgraph discovery algorithms are designed to capture only the first-order relationships.
Higher-order structures (i.e., motifs/graphlets) are shown to be the fundamental building blocks in the organization and dynamics of real-world networks such as social and neural networks~\cite{Milo02, Przulj07, SePiKo14, Pinar17, Ahmed16, Jain17}.
Discovering dense subgraphs that have a higher-order structure is imperative in the analysis of those networks.
In simple networks, motifs are used to find subgraphs with higher-order structure, which cannot be detected with edge-centric methods~\cite{Tsourakakis15, Sariyuce15}.
However, it is not clear how to find dense subgraphs with motifs in heterogeneous networks.
The spectrum of motifs in heterogeneous networks is wide due to the edge directions and node/edge attributes. 
Although this variety is particularly effective for the analysis as the structure and dynamics can vary with respect to the type of motif considered~\cite{Benson16, Babis17, sun2011pathsim}, handling the diverse nature of heterogeneous networks while watching for the motifs is a challenging problem.

In this work, we introduce the {\bf quark decomposition} framework to find motif-driven subgraphs in networks with directed edges and categorical attributes on nodes/edges.
Our framework builds subgraphs, called quarks, in varying quality and with hierarchical relations.
A $k$-quark is a motif-parameterized subgraph where each node/edge (or small motif) participates in a number of (larger) motifs. The parameter $k$ denotes the extent of participation and not an input: quark decomposition finds non-empty $k$-quarks for all $k$ values.
Our framework is inspired by the peeling approach in simple graphs, namely core, truss, and nucleus decompositions~\cite{Seidman83, Cohen08, Sariyuce15}, which first locate the outer sparse parts of the graph and then find the inner dense regions.
Given the practicality and effectiveness of the peeling techniques, we adapt them to the directed networks with categorical attributes in a principled way.
Note that it is beyond nontrivial to consider a generalization since edge directions and node/edge attributes create a wide and diverse spectrum of motifs (see~\cref{fig:dirmotifs} and~\cref{fig:signedmotifs} for examples).

Quark decomposition takes two motifs as parameters, $M$ and $N$ where $M \subset N$, and builds subgraphs where $M$s participate in many $N$s.
The parameterized formulation enables the discovery of diverse subgraphs and lets a trade-off between quality and practicality.
We assign density indicators for each motif $M$, called quark numbers, to denote how well $M$ is connected to its neighborhood, which are then used to build the $k$-quarks.
Quark numbers also relate the $k$-quarks with each other using containment relations; the quarks with larger $k$ are contained in the ones with smaller $k$.

We use quark decomposition in two broad classes of applications: (1) When the motif of interest is unknown and higher-order organization of the network is asked; (2) When the motif of interest is known and guides the quark decomposition.
For (1), we compare quarks obtained with various motifs and analyze the number and quality of resulting quarks in directed and signed-directed networks. In a word-association network, we show that overlapping quarks by the same motif as well as the quarks by different motifs capture diverse contexts of the words.
We also analyze the Florida Bay food web and show that quarks obtain consistently better results than the state-of-the-art algorithm for finding groups of compartments with respect to the ground-truth classifications.
For (2), we consider the task of finding gender-balanced communities in networks where genders are used as node labels. We focus on the college friendship networks that have low female ratios. We consider clique-based instantiations of the quark decomposition and choose gender-balanced triangles and four-cliques as the motif $N$. 
We observe consistently high female ratios when compared to the label-oblivious state-of-the-art methods.

\noindent Key contributions in this paper are summarized as follows:

\begin{compactitem} [\leftmargin=-.5in$\bullet$]
\item \textbf{Quark decomposition.} We propose a framework to find dense subgraphs according to a given motif in networks with directed edges and categorical attributes on nodes/edges.
Quark decomposition is versatile, efficient, and extendible.

\item \textbf{Limitations and role confusion problem.} We characterize the limitations and practical instantiations of quark decomposition to guide the selection of parameter motifs $M$ and $N$.
Presence of multiple orbits in some directed motifs results in subgraphs where a node/edge serves in multiple orbits. We call this role confusion and devise role-aware quark numbers as a remedy.

\item \textbf{Generic peeling algorithm for any motif.} We introduce a generic peeling algorithm that works for any motif pair $M, N$ to find the quark decomposition. Our algorithm is similar in spirit to the core/truss/nucleus decompositions and can enjoy the optimizations applicable for the peeling algorithms.

\item \textbf{Extensive evaluation on real-world networks.} We evaluate quark decomposition on three types of heterogeneous networks; directed, signed-directed, and node-labeled. We consider various motifs using several real-world networks.
Results show that quark decomposition outperforms the state-of-the-art techniques and is also practically scalable to networks with up to 101M edges.
\end{compactitem}

\vspace{-1ex}
\section{Related Work}\label{sec:rel}
Here we summarize the related works on motif-driven dense subgraphs in heterogeneous networks. Note that there are too many clustering, community detection/search works on heterogeneous networks~\cite{fang2020survey}, but our focus is limited to motif-based approaches that find dense subgraphs in directed and labeled networks.

\noindent {\bf Peeling approaches on simple networks.} Core decomposition is a simple but effective model to locate the seed regions where dense subgraphs can be found~\cite{Seidman83}.
It makes use of degrees to assign core numbers to nodes.
Regularizing the node degrees, that span to a large range, to the core numbers in a smaller range is the key and makes the peeling a fundamental building block in an array of applications~\cite{peng2014accelerating, Kitsak10, Malliaros16, Alvarez06, Carmi07}.
Higher-order variants of the peeling are introduced to take advantage of the triangles and small subgraphs.
Truss decomposition leverages triangles~\cite{Saito06, Cohen08, Huang14}, nucleus decomposition makes use of small cliques~\cite{Sariyuce15, Sariyuce-TWEB17}, and tip/wing decomposition utilizes butterflies ($2,2$-bicliques) in bipartite networks~\cite{Aksoy16, Sariyuce18} to find dense regions in a principled way.
Nucleus decomposition generalizes core and truss decompositions. It finds $k$-$(r,s)$ nucleus, defined as a subgraph where each $r$-clique is a part of $k$ number of $s$-cliques where $r<s$.
Those peeling approaches work on undirected simple networks. Here we work on directed networks with attributes on nodes/edges.
We compare our methods with nucleus decomposition to highlight the benefit of considering edge directions and node labels.

\noindent {\bf Cycle-, flow-trusses.} Regarding the adaptations of core and truss decompositions for directed graphs, Takaguchi and Yoshida~\cite{Takaguchi16} introduced cycle- and flow-trusses.
Their algorithms work on directed networks with respect to cycle and flow (\ita{acyclic}) motifs (see~\cref{fig:dirmotifs}) and rely on the occurrences of the cycle and flow motifs for each edge. However, they do not consider the bidirectional edges and handle each such edge as two separate unidirectional edges. We compare our work with cycle-flow trusses in~\cref{sec:expdirected}.

\noindent {\bf Motif-based densest subgraph and clique finding.} There are a few works in the literature that studies the motif-driven densest subgraph problem. For constant size cliques, which can be thought of as motifs in simple undirected graphs, Tsourakakis introduced the $k$-clique densest subgraph problem~\cite{Tsourakakis15} to generalize the classical densest subgraph discovery~\cite{Goldberg84} for $k$-cliques ($k>2$). Analogous to finding the subgraph with the maximum average degree, Tsourakakis proposed to find the subgraphs that have a maximum average triangle (or $k$-clique) count per node.
More recently, Fang et al. proposed exact and approximate algorithms for the same problem~\cite{fang2019efficient} and Hu et al. considered heterogeneous information networks with specific schemas to find maximal motif-cliques~\cite{hu2019discovering}.
Note that our problem setup is more general and we aim to find multiple subgraphs that are not necessarily perfect cliques but always significantly dense.

\noindent {\bf Higher-order motif clustering.} The motif-based graph clustering problem is studied from a spectral perspective in heterogeneous networks. Benson et al.~\cite{Benson16} introduced a nice generalized framework for clustering the networks based on the higher-order connectivity patterns. They defined the motif conductance as the ratio of the number of motifs cutting the border between two regions to the number of motif instance endpoints (i.e., nodes) in the subgraph or its complement, whichever is smaller. Since getting the optimal solution for motif conductance is NP-Hard, they proposed an approximate algorithm that finds the near-optimal cluster in a given network. Their method relies on the spectral clustering of motif adjacency matrix whose entry $i, j$ is the number of motifs where nodes $i$ and $j$ co-occur. The set of nodes in the spectral ordering that has the minimum conductance is reported as the optimal higher-order cluster. 
 Recursive bisection method iteratively finds the near-optimal clusters in the complement of the graph and k-means clustering on the motif adjacency matrix obtains a prespecified number of clusters.
Concurrently, Tsourakakis et al.~\cite{Babis17} proposed the same framework for motif-aware clustering and also introduced a random walk interpretation of the graph reweighting scheme which gives a principled approach to define the notion of conductance for other motifs.
More recently, Li et al. improved the motif-based clustering approach to handle the clustering of disconnected nodes~\cite{li2019edmot}.
Our framework differs from those approaches: we do not partition the graph but find dense subgraphs around nodes/edges.
We give an extensive comparison against the motif clustering~\cite{Benson16}  in our experiments.

\begin{figure}[!t]
\centering
\includegraphics[width=0.8\linewidth]{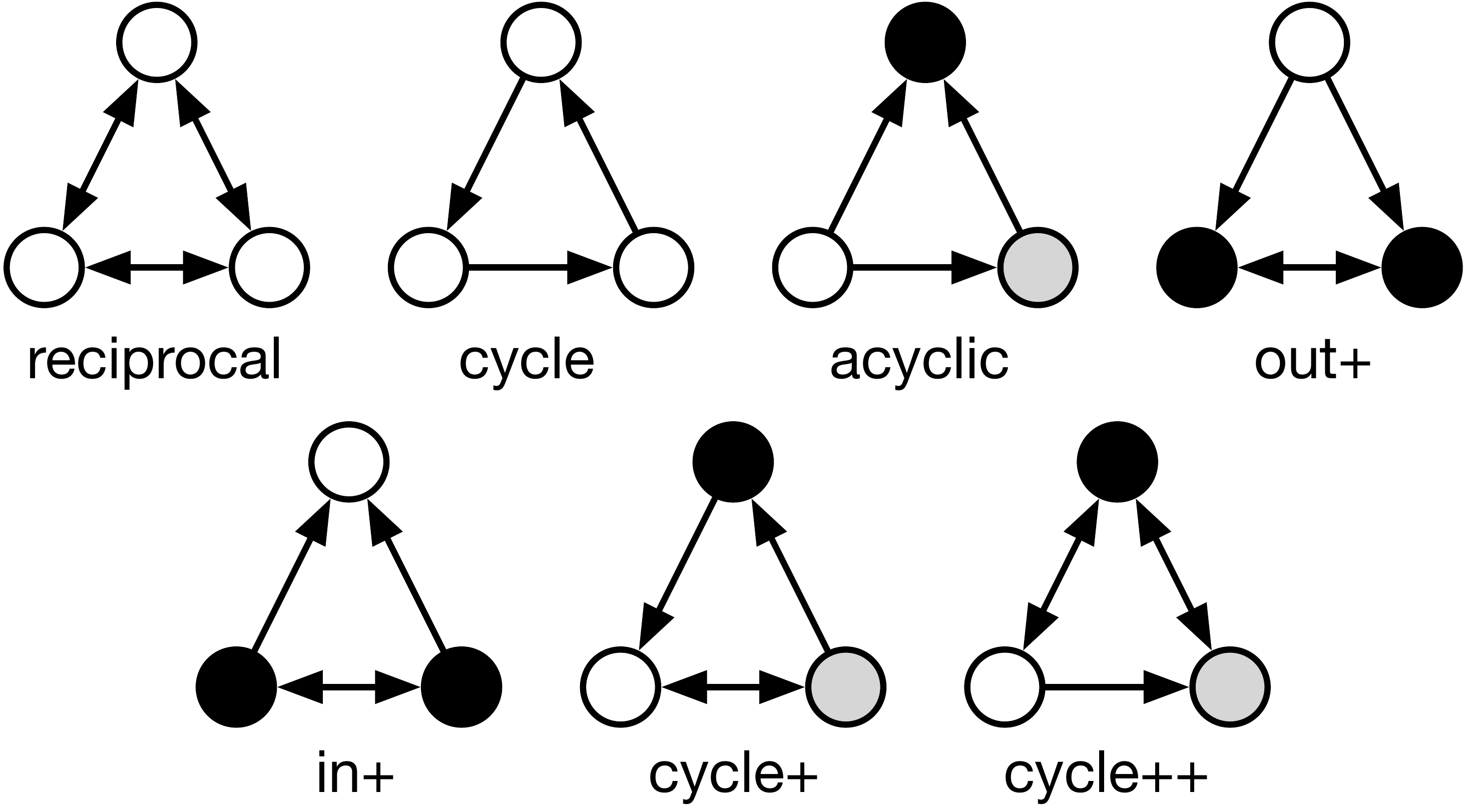}
\vspace{-3ex}
\caption{\small Directed triangle motifs with three nodes and three edges, as named in~\cite{Sesh16}. Automorphism orbit (or just orbit) of a node in a given motif is the set of other nodes that have the same topological connectivity patterns. In each triangle, orbits of the nodes are denoted with colors (i.e., nodes in the same color have the same orbit).}
\label{fig:dirmotifs}
\end{figure}

\section{Preliminaries}\label{sec:prelim}

\noindent {\bf Motifs and hypergraphs.} We define motif $M=(V_M, E_M)$ as an induced directed subgraph with node and edge sets $V_M$, $E_M$. Each $v \in V_M$ and $e \in E_M$ can have categorical (non-numeric) attributes, defined by $f:V_M, E_M\rightarrow \mathcal{N}$. 
A motif $M$ is a subset of motif $N$ iff
\begin{compactitem} [\leftmargin=-.5in$\bullet$]
    \item  $V_M \subseteq V_N$ where there is one and only one $v' \in V_N$ for each $v \in V_M$ such that $f(v)=f(v')$.
    \item $E_M \subseteq E_N$ where there is one and only one $e' \in E_N$ for each $e \in E_M$ such that $f(e)=f(e')$.
\end{compactitem}

\noindent We use the language of \emph{hypergraphs} to define the involvements of small motifs in the larger motifs.
A hypergraph $H = (V,E)$ consists of the node set $V$ and hyperedge set $E$, where a hyperedge $e \in E$ is simply a subset of $V$ (in standard graphs, each hyperedge has two nodes).
Consider a hypergraph $H=(V,E)$;

\begin{compactitem}[\leftmargin=-.15in$\bullet$]
    \item $u, v \in V$ are neighbors if there is a hyperedge $e \in E$ that contains $u$ and $v$.
    \item The degree of a node $v \in V$, denoted by $d(v)$, is the number of hyperedges that contain $v$.
    \item The size of a hyperedge $e \in E$, denoted by $s(e)$, is the number of nodes in it.
    \item Two nodes $u$ and $v$ are connected if there exists a sequence of hyperedges
    $e_1, e_2, \ldots, e_\ell \in E$ such that $u \in e_1$, $v \in e_\ell$, and $\forall i < \ell$,
    $e_i \cap e_{i+1} \neq \emptyset$.
    \item $H$ is connected if all pairs of nodes are connected.
\end{compactitem}
 
\begin{definition} \label{def:induced}
Let $S \subseteq V$.
The \textbf{induced hypergraph} $H|_S$ has node set $S$ and contains every hyperedge of $H$ completely contained in $S$, i.e., $\forall e \in H|_S$, iff $v \in e$ then $v \in S$.
\begin{compactitem}[\leftmargin=-.15in$\bullet$]
    \item The degree of node $v$\,$\in$\,$H|_S$ is denoted by $d_S(v)$ (or $d(v)$ when clear).
    \item The minimum degree in $H|_S$ is denoted by $\delta_S$.
\end{compactitem}
\end{definition} 

\noindent Induced hypergraph is also known as section hypergraph.

\noindent {\bf Dense subgraphs.} 
We call a subgraph dense if it has many motifs (also called {\bf motif-based or -driven dense subgraph}). Formally, we use {\bf average motif degree} to quantify the density of a subgraph with respect to a given motif.

\begin{definition} \label{def:amd}
For a subgraph $S$ and a motif $N$, the {\bf average motif degree} of $S$ is the number of instances of $N$ in $S$ divided by the number of nodes in $S$.
\end{definition} 

\section{Quark decomposition framework}\label{sec:frame}

We first define the {\bf motif hypergraph}.

\begin{definition} \label{def:motifhypergraph} 
Given a graph $G$ and template motifs $M$ and $N$ s.t. $M \subset N$.
Let $\{M\}$ and $\{N\}$ be the set of instances of $M$ and $N$ in $G$, respectively, and $f, g$ be bijective functions.
{\bf Motif hypergraph} $H_G(M, N)=(V_G, E_G)$ is a hypergraph constructed as follows:
\begin{compactitem}[\leftmargin=-.15in$\bullet$]
\item Each instance of $M\in G$ forms a node $u \in V_G$ by $f$:$\{M\}\rightarrow {V_G}$.
\item Each instance of $N\in G$ forms a hyperedge $e \in E_G$ by $g$:$\{N\}\rightarrow {E_G}$.
\item Iff $M\subset N$ in $G$, then $f(M)\in g(N)$ in $H_G$.
\end{compactitem}
\end{definition} 

\noindent Note that $H_G(M,N)$ is a $t$-uniform hypergraph (i.e., $s(e)=t~~\forall e\in E_G$) where $t$ is the number of occurrences of $M$ in $N$.
We also refer to the degree of each node $u \in H_G$ as the {\bf motif degree}, denoted by $d_{H_G}(u)$ (or $d(u)$ when $H_G$ is obvious).

\noindent We now introduce the notion of \emph{k-quark} subgraph.

\begin{definition} \label{def:quark}
Given a graph $G$ and template motifs $M$, $N$ such that $M \subset N$, say $H_G(M,N)$ is the motif hypergraph defined as above.
\begin{compactitem}[\leftmargin=-.15in$\bullet$]
    \item For any $k \in \NN$, a \textbf{\mb{k}-quark} of $H_G(M,N)$ is a \emph{connected} and \emph{maximal} induced sub-hypergraph $H|_S$ such that $\delta_S \geq k$.
    \item For a node $u$ in $H_G(M,N)$ (corresponding to an instance of motif $M$ in $G$), the \textbf{quark number} of $u$, denoted by $\quark{u}$, is the largest value of $k$ such that $u$ belongs to a non-empty $k$-quark.
\end{compactitem}    
\end{definition}

\noindent We also refer to $k$-quark as quark when $k$ is irrelevant or clear.

\begin{definition} \label{def:hier}
$k$-quarks form a hierarchy by containment. 
\begin{compactitem}[\leftmargin=-.15in$\bullet$]
\item Let $S$ be a $k$-quark and $T$ be a $k'$-quark such that $k'<k$ and $S \subset T$. $S$ is the \textbf{child} of $T$ (and $T$ is the {\bf parent} of $S$) if there is no $\bar{k}$-quark $U$ such that $k'<\bar{k}<k$ and $S \subset U \subset T$.
\item A $k$-quark is {\bf leaf} (childless) if there is no $k^+$-quark in it s.t. $k^+>k$.
\item {\bf Maximum quark number} of a graph is the largest $k$ for which there is a non-empty $k$-quark.
\item {\bf Maximum $k$-quark} is a quark where $k$ is the maximum quark number in the graph.
\end{compactitem}    
\end{definition}

\textit{Quark decomposition} is the process of finding the quark numbers and $k$-quarks for a given pair of motifs $M, N$ in a graph $G$. Leaf $k$-quarks are the locally optimal subgraphs; they are surrounded by less dense quarks (with lower $k$ values) and hence often contain the most interesting information. 

If $G$ is a simple undirected graph where $M$ is $r$-clique and $N$ is $s$-clique ($r<s$), then $k$-quark is nothing but a $k$-$(r,s)$ nucleus~\cite{Sariyuce15, Sariyuce-TWEB17}.
If $G$ is a simple undirected bipartite graph where $M$ is edge and $N$ is $2,2$-biclique, then $k$-quark reduces to be a $k$-wing~\cite{Sariyuce18}.
In the quark decomposition, each $M$ instance is given a quark number and the $k$-quarks along with the hierarchical relationships among them can be constructed accordingly.
{\bf Note that, $N$ is the motif of interest for which dense regions are to be found.}
Motif $M$ can be any subset of $N$ but it should satisfy some requirements such that non-trivial $k$-quarks can be obtained (more details given in~\cref{sec:limits}).
When $M$ is edge or a larger motif, the resulting $k$-quarks may overlap with each other because quarks are defined as a group of $M$s.
This is useful since the overlapping communities can better capture the network organization~\cite{Xie13}.

\begin{figure}[!t]
\includegraphics[width=\linewidth]{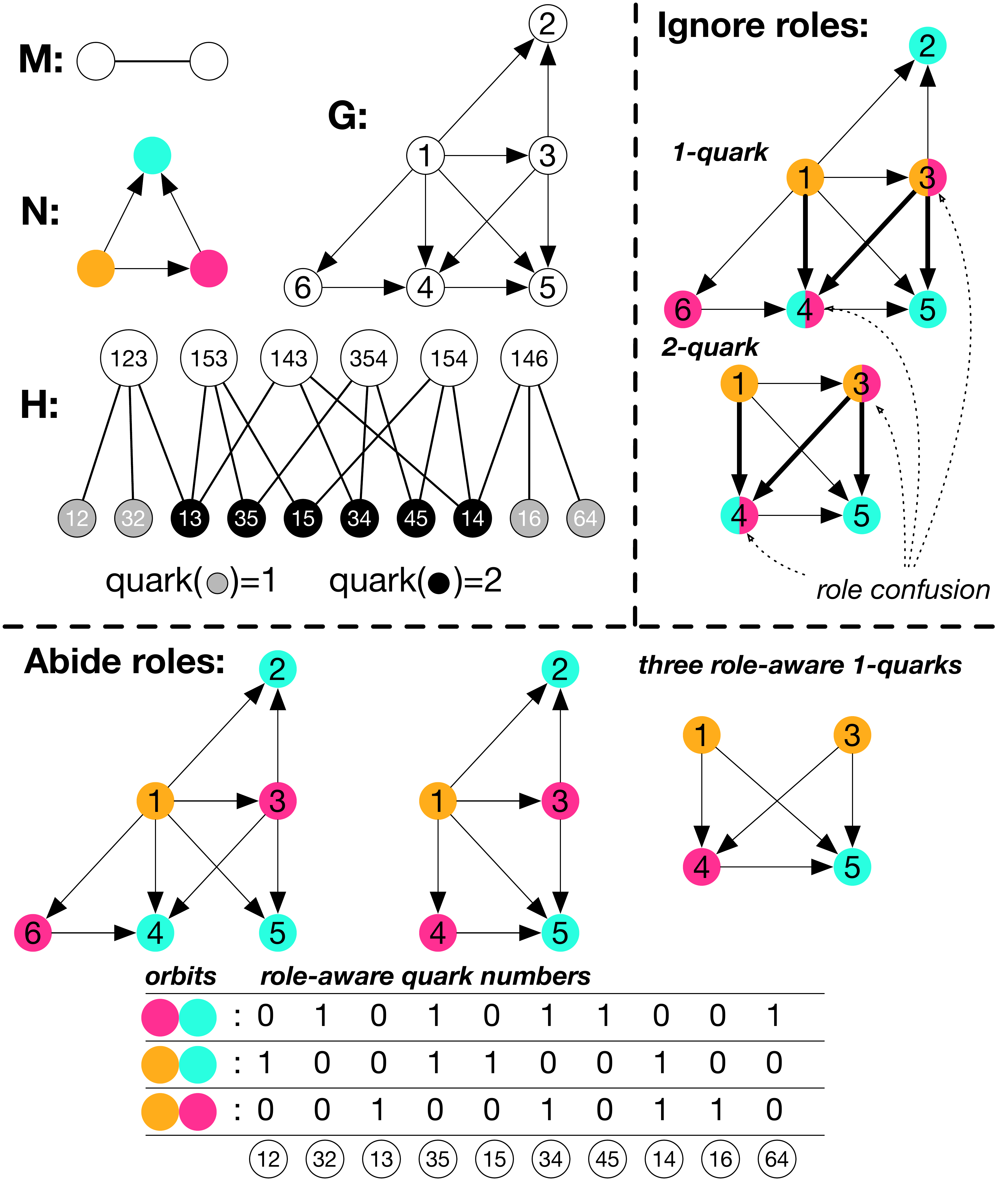}
\vspace{-5ex}
\caption{\small We construct the motif hypergraph \mb{H} with respect to the motifs \mb{M} (edge) and \mb{N} (acyclic) on a toy graph \mb{G} (in top-left). Each orbit in \mb{N} is shown by a different color. We create a node for each motif \mb{M} and a hyperedge for each motif \mb{N} to get the motif hypergraph \mb{H}. Ids of triangles and edges are the union of the nodes in each. Quark numbers are denoted with gray for 1 and black for 2. 1-quark and 2-quark are shown in top-right. Role confusion occurs for nodes 3 and 4 in both quarks. Role-aware quarks and quark numbers are shown at the bottom (See~\cref{sec:example} for more details).}
\label{fig:role}
\end{figure}

\vspace{-1ex}
\subsection{Limitations \& practical instantiations}\label{sec:limits}

It is important to consider the necessary and sufficient conditions for the $M$ and $N$ in~\cref{def:quark} so that the $k$-quarks are non-trivial. For instance, if $M$ has only one occurrence in $N$, the size of each hyperedge in the motif hypergraph becomes one, thus it is not possible to consider any connectivity among $M$s.
This is related to the automorphism orbits~\cite{Przulj07, Sarajlic16}.
Automorphism orbit (or just orbit) of a node in a given motif is the set of other nodes that have the same topological connectivity patterns.
For directed triangle motifs, shown in~\cref{fig:dirmotifs}, orbits of the nodes are denoted with colors.
For instance, \ita{in+} has two orbits; the first has one node with two incoming edges (white node) and the second has two nodes where each has one outgoing and one bidirectional edge (black nodes).
Note that automorphism orbits are defined only for the nodes.
In our framework, we can also consider an edge or a larger structure as $M$ in~\cref{def:quark}, thus automorphism orbits of such structures need to be taken into account.
To define the orbits of any $M$, we consider the ordered list of node orbits in it.

Enforcing the automorphism orbits restricts the use of any motif since~\cref{def:quark} requires that at least one orbit should have multiple instances in $N$.
E.g., when $M$ is node or edge; \ita{acyclic}, \ita{cycle+}, and \ita{cycle++} motifs have three different orbits, i.e., no orbit has multiple members. Thus those cannot be considered as $N$.
To remedy this problem, \textbf{we use a vanilla motif $M$, which has only one orbit}. A vanilla motif has no  node/edge labels and its edges are directionless\footnote{we do not say bidirectional to avoid any confusion}. Any motif $N \supset M$ will contain multiple instances of $M$, thus the size of hyperedges will be greater than one.
\cref{fig:role} gives an example when $M$ is a vanilla edge and $N$ is \ita{acyclic}.

\subsubsection{\bf Role confusion problem}\label{sec:roleconf}

Another problem in $k$-quarks is the conflation of the orbits for $M$s.
An instance of $M$ can be a part of multiple instances of $N$.
Furthermore, orbit of an $M$ instance in one $N$ instance can be different than its orbit in another $N$ instance. E.g., if $M$ is node and $N$ is \ita{out+}, a node may appear as white node in one \ita{out+} instance while being black node in another \ita{out+} instance (see~\cref{fig:dirmotifs}).
We call this {\it the role confusion problem}.
Note that choosing $M$ as a vanilla motif does not help with the role confusion problem.
For example, consider \ita{cycle+} with nodes $A, B, C$ and edges $A$$\rightarrow$$B$, $B$$\rightarrow$$C$, $A$$\leftrightarrow$$C$.
This can be a structure in Twitter network; e.g., A is a grad student, C is her advisor (a professor), and B is a junior faculty working in the same field---B is an interesting person for A but not for C, professor C is a well-known person followed by many and she follows A since she is A's advisor.
In the motif-driven subgraphs for \ita{cycle+}, each node has ideally a single role; i.e., a grad student is better characterized as the node A in all the \ita{cycle+} instances she participates in. 
The solution is to construct the subgraphs in a way that abide by the orbits.
To do that, we define {\bf role-aware quark numbers} for each $M$.
If $M$ has $b$ orbits in $N$, each $M$ will have $b$ role-aware quark numbers, one for each orbit.

\begin{definition} \label{def:RAquark}
Given a graph $G$ and motifs (vanilla) $M$ and $N$ (s.t. $M \subset N$), let $H_G(M,N)$ be the motif hypergraph as defined in~\cref{def:motifhypergraph}. Let $b$ be the number of orbits of $M$ in $N$.
\begin{compactitem}[\leftmargin=-.15in$\bullet$]
    \item \textbf{Orbit degree} of an $M$ instance is the number of instances of $N$ that contain it such that the orbit of the $M$ instance in each $N$ instance is the same. Each $M$ instance has $b$ orbit degrees.
    \item For any $k \in \NN$, a \textbf{role-aware \mb{k}-quark} of $H_G(M,N)$ is a \emph{connected} and \emph{maximal} induced sub-hypergraph $H|_S$ such that each $M$ instance has one orbit degree of at least $k$.    
    \item For a node $u$ in $H_G(M,N)$ (corresponding to an instance of motif $M$ in $G$), the \textbf{role-aware quark numbers} of $u$ are $b$ numbers, denoted by $\quarkRole{i}{u}$ for $1\le i \le b$.
    $\quarkRole{i}{u}$ is the largest value of $k$ such that the node $u$ belongs to a non-empty role-aware $k$-quark where its orbit is $i$.
\end{compactitem}    
\end{definition}

In a role-aware $k$-quark, each $M$ instance participates in at least $k$ $N$ instances and the orbit of the $M$ instance in each of those participations is the same (note that different $M$ instances can have different orbits in the quark).
Role-aware quark numbers describe the extent of participation as a particular orbit while the quark number indicates the extent of participation as {\it any} orbit.
When $M$ is node and $N$ is \ita{cycle} or \ita{reciprocal} (in~\cref{fig:dirmotifs}), there is no role confusion since there is only one node orbit in each. For all the other directed triangles, there is a role confusion.
When $M$ is edge, we always consider the orbits of the two nodes in it.
In the context of directed triangles, one important point is the distinction between unidirectional and bidirectional edges.
In this work, we consider the bidirectional edge as an atomic entity; i.e., not combination of two unidirectional edges, due to the fact that a symmetrical relation has a different semantic than two asymmetrical ones---in a sense we only consider the induced edges.
For instance, \ita{in+} and \ita{out+} do not create any role confusion when $M$ is edge, because the unidirectional edge between the white and black nodes in \ita{in+} (or \ita{out+}) cannot serve as a bidirectional edge (between two black nodes) in an adjacent \ita{in+} (or \ita{out+}) motif (see~\cref{fig:dirmotifs}).

\subsubsection{\bf Example}\label{sec:example}
\cref{fig:role} illustrates an example on a toy graph.
We choose the vanilla edge as the motif $M$ and \ita{acyclic} triangle as the motif $N$, shown in the top-left.
We denote the node orbits in acyclic with different colors; orange shows the node with two outgoing edges, blue is for the node with two incoming edges, and red denotes the node with one incoming and one outgoing edge.
Motif hypergraph, $H$, with respect to $G, M, N$ is shown next. 
As described in~\cref{def:motifhypergraph}, we create a node for each vanilla edge (bottom set in $H$) and a hyperedge for each acyclic (top set in $H$).
Id of each acyclic and edge in $H$ is formed by the concatenation of the node ids in it, e.g., $12$ denotes the edge from $1$ to $2$.
Each hyperedge is connected to three nodes since there are three edges in an acyclic---making $H$ a $3$-uniform hypergraph.
The degree of each node in $H$ is called the motif degree, e.g., it is $2$ for edge $4$-$5$.
Quark numbers of the edges are denoted by gray ($\mathcal{K}=1$) and black ($\mathcal{K}=2$) in the bottom set of $H$; e.g., quark number of edge $4$-$5$ is $2$.
Based on those quark numbers, we construct the quarks in the top-right.
Two quarks are created; a $1$-quark, corresponding to the entire graph $G$, and a $2$-quark with nodes $1, 3, 4,$ and $5$.
We observe role confusions for the thick edges in those quarks, which in turn implies the role confusions on nodes.
Regarding the edges, $1$-$4$, $3$-$4$, and $3$-$5$ have role confusion in both quarks; e.g., in 1-quark, $1$-$4$ connects orange to blue in $1$-$4$-$6$ acyclic but it links orange to red in $1$-$5$-$4$ acyclic.
For the nodes, $3$ and $4$ in both quarks have role confusions, e.g., node $3$ is red in the $1$-$3$-$5$ acyclic (and in $1$-$3$-$4$) while being orange in the $3$-$4$-$5$ acyclic.
We give the role-aware quarks and quark numbers of the edges at the bottom.
If we construct the role-aware quarks by abiding the orbits, we get three role-aware 1-quarks.
Note that the only edges that have multiple non-zero $\mathcal{K}_i$ are the ones that have role confusion in the top-right.

\subsection{Algorithms}\label{sec:alg}

Here we discuss our peeling algorithms, first for quark decomposition (\cref{def:quark}) and then for role-aware quarks (\cref{def:RAquark}).

\subsubsection{\bf Quark decomposition}\label{sec:quarkalg}

\noindent Algorithm~\ref{alg:quark} outlines the quark decomposition. Here we assume the motif $M$ in~\cref{def:quark} to be node or edge for simplicity. Note that larger structures can be considered as well. Also, we avoid constructing the actual hypergraph of motifs since it requires enumerating all the $N$s which will have a significant space cost. Instead, we discover those motifs for each node/edge as needed, similar to the space-efficient approach in nucleus decomposition~\cite{Sariyuce15}. There are three different phases in the quark decomposition; motif degree counting, peeling to find quark numbers, and constructing the subgraphs and the hierarchy.

\noindent \textbf{Motif degree counting.} Line~\ref{ln:one} corresponds to counting the occurrences of $N$ instances per each node or edge.
There are some studies~\cite{Ine18, Pinar17, Yaveroglu14} that leverage certain commonalities among motifs to avoid redundant computations.
Adapting those studies for per node/edge motif counting to improve runtime is possible but out of scope for this work.
Also, simultaneously counting the motif degrees for \textit{multiple motifs} would speed up the workflow when the motif of interest is unknown to the user and multiple options need to be investigated.
Note that if the label set of the input motif $N$ is smaller than the label set of the input graph $G$, we can filter the graph to only keep the label set of $N$, e.g., if a triangle of female students is $N$ for an undirected graph where genders are the node labels, only the induced graph of female nodes can be considered.

\noindent \textbf{Peeling to find quark numbers.} Lines~\ref{ln:three} to~\ref{ln:ten} assigns a quark number for each node/edge.
The classical core decomposition implementation~\cite{BaZa03} makes use of the bucket data structure to keep track of the nodes with the minimum degree.
We also use this approach to watch the nodes/edges with the minimum motif degree.
Initially, all the nodes/edges are marked as unprocessed (line~\ref{ln:two}), which will come handy to ensure each $N$ instance is processed only once.
In each iteration, an unprocessed node/edge with the minimum motif degree is chosen (line~\ref{ln:three}).
The motif degree of the chosen node/edge is set to be its quark number ($\mathcal{K}$) (line~\ref{ln:four}).
Then, the $N$ instances that contain the chosen node/edge are found and processed in lines~\ref{ln:five} to~\ref{ln:nine}.
In each such instance of $N$, we make sure that the other nodes/edges (the neighbors of the chosen node/edge) are unprocessed (line~\ref{ln:six}); this is done to ensure that each $N$ instance is examined only once.
Then, the neighbor nodes/edges in each $N$ instance are checked (line~\ref{ln:seven}) and their motif degree is decremented if larger than the motif degree of the chosen node/edge (lines~\ref{ln:eight} and~\ref{ln:nine}).
At the end of the iteration, the chosen node/edge is marked as processed (line~\ref{ln:ten}).
For any $M, N$ motif pair, two basic procedures are necessary and sufficient to instantiate the quark decomposition;

\begin{compactitem}[\leftmargin=-.15in$\bullet$]
    \item Finding the $N$ instances that contain a given $M$ instance (line~\ref{ln:five}),
    \item Finding the $M$ instances in a given $N$ instance (lines~\ref{ln:six} and~\ref{ln:seven}).
\end{compactitem}

\noindent \textbf{Constructing the subgraphs and the hierarchy.} It is also possible to construct the subgraphs and build hierarchy among $k$-quarks during the peeling operation, as noted in magenta lines after lines~\ref{ln:six} and~\ref{ln:ten}.
Subgraph and hierarchy construction during the peeling process is introduced in~\cite{Sariyuce-VLDB16} and can be adapted to the quark decomposition.
In line~\ref{ln:six}, the nodes/edges that are in the same $N$ instance with the chosen node/edge are checked.
If those nodes/edges are already processed (i.e., assigned a quark number), we can build non-maximal $k$-quarks during the peeling process.
At the end, those non-maximal $k$-quarks are converted to the real (maximal) $k$-quarks with a light-weight post-processing that uses the union-find data structure (following line~\ref{ln:ten}).

\begin{algorithm}[!t]
\small
\DontPrintSemicolon
\caption{\textsc{Quark Decomposition} ($G\,(V,E)$, motif $N$)}
\label{alg:quark}
Compute $d(u)$ (motif deg.) $\forall~u \in V$ ~\tcp{\mycommfont{or $d(e)$ $\forall~e\in E$}}\label{ln:one}
Mark every $u$ {\mycommfont{(or $e$)}} as unprocessed\label{ln:two}\;
\For{\meach unprocessed $u$ {\mycommfont{(or $e$)}} with minimum degree $d$}{\label{ln:three}
	$\mathcal{K}(u) \leftarrow d(u)$~\tcp{\mycommfont{or $\mathcal{K}(e) \leftarrow d(e)$}}\label{ln:four}
	\For{\meach motif $N$ s.t.~~$u \subset N$ {\mycommfont{(or $e \subset N$)}}}{\label{ln:five}
    		\lIf{any node $v \in N$ is processed {\mycommfont{(or edge)}}}{continue}\label{ln:six}
			\tcp{\textcolor{magenta}{can also find subgraphs \& hierarchy}}
    		\For{\meach node $v~(\neq u) \subset N$ {\mycommfont{(or edge $f~(\neq e)$)}}} {\label{ln:seven}
	    		\If{$d(v) > d(u)$ {\mycommfont{(or $d(f) > d(e)$)}}}{\label{ln:eight} 
    				$d(v) \leftarrow d(v) - 1$~\tcp{\mycommfont{or $d(f) \leftarrow d(f) - 1$}}\label{ln:nine}
	    		}
    		}
    	}
  	Mark $u$ {\mycommfont{(or $e$)}} as processed\label{ln:ten}\;
  }
  \tcp{\textcolor{magenta}{Optional post-processing to build the hierarchy}}
  \Return array $\mathcal{K}(\cdot)$ \label{ln:eleven}\;
\end{algorithm}

\begin{theorem}
Given a graph $G=(V,E)$ and motif $N$, Algorithm~\ref{alg:quark} finds the quark numbers, $\quark{\cdot}$, of all $u \in V$ (or $e \in E$).
\end{theorem}
\begin{proof}
We give the proof for the node case without loss of generality (i.e., it is similar for the edge).
As noted in~\cref{sec:prelim}, nodes $u, v \in V$ are neighbors if they participate in a common motif $N$.
$\quark{u}=k$ indicates that there are at least $k$ instances of $N$ which contain $u$ and in each such $N$, $u$ has a neighbor node $v$ s.t. $\quark{v} \ge \quark{u}$.
This is enforced by the lines~\ref{ln:eight}-\ref{ln:nine}, where the motif degree a neighbor node is decreased if it is larger than the quark number assigned at that step.
In other words, any neighbor node with a smaller quark number does not contribute to the quark number of the node of interest.
If Algorithm~\ref{alg:quark} finds $\quark{u}=k$ for a node $u \in V$, then by~\cref{def:quark}, we need to show that (i) $\exists$ a $k$-quark $G' \ni u$, (ii) $\nexists$ a $k^+$-quark $G' \ni u$ ($k^+ > k$).\\
(i) Once $\quark{u}=k$ is found in Algorithm~\ref{alg:quark}, we stop and construct an induced subgraph $G' \subset G$ by traversal as follows.
Initially, $G'$ has only $u$.
In each step, we visit a node $v \in V$ s.t. $v$ co-participates in some $N$ instance with a node from $G'$.
If $\quark{v}=k$ or if it is unassigned but its current motif degree is equal to $k$, we add $v$ to $G'$.
We continue the traversal until no such node $v$ can be found.
At the end, $G'$ is a $k$-quark since (1) each node participates in $\ge k$ motifs, because the nodes are processed in the non-decreasing order of their motif degrees, (2) all the nodes are connected to each other with motifs due to the motif-based traversal, and (3) $G'$ is maximal since it is the largest subgraph that can be found by the traversal.\\
(ii) $u$ cannot be in a $k^+$-quark. Assume it is. Then it should take part in at least $k^+$ motifs and each motif contains a neighbor node with motif degree of at least $k^+$. But, this implies that $\quark{u}=k^+$ (by~\cref{def:quark}), contradiction.
\end{proof}

\noindent \textbf{Time and space complexity.} Algorithm~\ref{alg:quark} has $O(\sum_{v\in V} d(v)^{|V_N|-1})$ complexity when $M$ is node or edge ($V_N$ is the node set of $N$).
The space complexity is $O(|E|)$ when $M$ is node/edge.
Instead of explicitly building the hypergraph of $M$s and $N$s in $G$, we only build the adjacency lists when required.
Since $N$s are not stored, space complexity is bounded by $O(|E|)$.
We find all motifs containing the node/edge of interest only when that node/edge is processed.
Each node/edge is processed at most once.
When $M$ is node, we can find all the $N$s containing a node by looking at all $(|V_N|-1)$-tuples in each of the neighborhoods of the node.
This takes at most $\sum_{v \in V} d(v)^{|V_N|-1}$.
Likewise, when $M$ is edge, we consider $|V_N|-2$ tuples in each edge neighborhood and total time is $\sum_{e \in E} \sum_{v \in e} d(v)^{|V_N|-2} $ $= \sum_{v \in V} \sum_{e \ni v} d(v)^{|V_N|-2} $ $= \sum_v d(v)^{|V_N|-1}$.

\begin{algorithm}[!t]
\small
\DontPrintSemicolon
\caption{\textsc{Role-aware Quark Dec.}\,($G\,(V,E)$,\,motif\,$N$)}
\label{alg:RAquark}
Let $B$ be the set of orbits that a node/edge has in $N$\label{ln:RAone}\;
Compute\,$d_o(u)$\,(orbit\,deg.)\,$\forall$ orbits $o$=$1,...,|B|$,\,$\forall\,u$$\in$$V$\,\tcp{\mycommfont{or\,$d_o(e)$}}\label{ln:RAtwo}
Mark every tuple $(u, o)$ as unprocessed for $o$=$1,...,|B|$~\tcp{\mycommfont{or $(e,o)$}}\label{ln:RAthree}
\For{\meach unprocessed  $(u, a)$ {\mycommfont{(or $(e,a)$)}} with min. orbit degree $d_a$}{\label{ln:RAfour}
	$\mathcal{K}_a(u) \leftarrow d_a(u)$~\tcp{{\mycommfont{or $\mathcal{K}_a(e) \leftarrow d_a(e)$}}}\label{ln:RAfive}
	\For{\meach motif $N$ s.t.~~$u \subset N$ {\mycommfont{(or $e \subset N$)}}}{ \label{ln:RAsix}
		Let\,$v$($\ne$$u$)\,be\,a\,node\,in\,$N$,\,$b$\,be\,its\,orbit\,\mycommfont{(or\,$f$($\ne\,e$)\,is\,an\,edge)}\label{ln:RAseven}\;
    		\lIf{any\,tuple\,$(v, b)\,\in\,N$\,is\,processed\,{\mycommfont{(or\,$(e,b)$)}}}{cont.}\label{ln:RAeight}
    		\For{\meach node $v~(\neq u) \subset N$ {\mycommfont{(or edge $f~(\neq e)$)}}} {\label{ln:RAnine}
			Let $b$ be the orbit of $v$ {\mycommfont{(or $f$)}} in $N$\label{ln:RAten}\;
	    		\If{$d_b(v) > d_a(u)$ {\mycommfont{(or $d_b(f) > d_a(e)$)}}}{\label{ln:RAeleven} 
    				$d_b(v) \leftarrow d_b(v) - 1$~\tcp{\mycommfont{or $d_b(f) \leftarrow d_b(f) - 1$}}\label{ln:RAtwelve}
	    		}
    		}
    	}
  	Mark $(u,a)$ {\mycommfont{(or $(e,a)$)}} as processed\label{ln:RAthirteen}\;
  }
  \Return arrays $\mathcal{K}_1(\cdot), ..., \mathcal{K}_{|B|}(\cdot)$ \label{ln:RAfourteen}\;
\end{algorithm}

\subsubsection{\bf Role-aware quark decomposition}\label{sec:quarkalg}
Algorithm~\ref{alg:RAquark} outlines the role-aware quark decomposition; again, for simplicity, we assume the motif $M$ in~\cref{def:RAquark} is node or edge (larger $M$ can be considered as well).
The only difference with respect to Algorithm~\ref{alg:quark} is the way we compute and keep the degrees and process the node/edge in the inner loop.
We first find the set of orbits, $B$, that a node/edge has in $N$ (line~\ref{ln:RAone}).
In line~\ref{ln:RAtwo}, we count the orbit degrees for each node/edge and for all orbits $1,...,|B|$.
Orbit degree $d_o$ of a node $u$ (or edge $e$) is the number of $N$ instances that contain the node $u$ (edge $e$) such that the orbit of $u$ ($e$) is $o$ (\cref{def:RAquark}).
Each orbit degree of a node/edge is processed separately. 
We also keep $|B|$ arrays to keep track of the processed $(u, o)$ (or $(e,o)$) tuples ($o$ is the orbit of node $u$, or edge $e$) (line~\ref{ln:RAthree}).
In lines~\ref{ln:RAfour}-\ref{ln:RAthirteen}, we process all the orbit degrees in non-decreasing order.
Role-aware quark number is assigned for the chosen node/edge (line~\ref{ln:RAfive}) and we find the neighbors of the node/edge in each $N$ to adjust their orbit degrees (lines~\ref{ln:RAsix} to~\ref{ln:RAtwelve}).
At the end, we return role-aware quark numbers for each node; $\mathcal{K}_i(\cdot)$ for $1\le i \le |B|$.

\begin{table*}[!t]
\centering
\small
\caption{\small Directed datasets from various domains. $|V|$, $|E|$, $|E_u|$, and $|E_d|$ are the number of nodes, edges, bidirectional, and unidirectional edges. We also list the number of motifs for each directed triangle (see~\cref{fig:dirmotifs}) and maximum quark numbers in each quark decomposition. Largest motif count and quark number for each network are shown in bold (for quark numbers, $M$ is a vanilla edge, $N$ is a triangle from~\cref{fig:dirmotifs}).}
\vspace{-3ex}
\renewcommand{\tabcolsep}{1pt}
\hspace{-6ex}
\begin{tabular}{|c|r|r|r|r||r|r|r|r|r|r|r||r|r|r|r|r|r|r|}\hline
& \multicolumn{1}{c|}{\multirow{2}{*}{$|V|$}} & \multicolumn{1}{c|}{\multirow{2}{*}{$|E|$}} & \multicolumn{1}{c|}{\multirow{2}{*}{$|E_u|$}} & \multicolumn{1}{c||}{\multirow{2}{*}{$|E_d|$}} & \multicolumn{7}{c||}{\# motifs} & \multicolumn{7}{c|}{maximum quark numbers}  \\
& & & & & \multicolumn{1}{c|}{\ita{cycle}} & \multicolumn{1}{c|}{\ita{acyclic}} & \multicolumn{1}{c|}{\ita{out+}} & \multicolumn{1}{c|}{\ita{in+}} & \multicolumn{1}{c|}{\ita{cycle+}} & \multicolumn{1}{c|}{\ita{cycle++}} & \multicolumn{1}{c||}{\ita{reciprocal}} & \multicolumn{1}{c|}{\ita{cycle}} & \multicolumn{1}{c|}{\ita{acyclic}} & \multicolumn{1}{c|}{\ita{out+}} & \multicolumn{1}{c|}{\ita{in+}} & \multicolumn{1}{c|}{\ita{cycle+}} & \multicolumn{1}{c|}{\ita{cycle++}} & \multicolumn{1}{c|}{\ita{reciprocal}}\\ \hline
\ti{foodweb}	&	128	&	 2.1K	&	31	&	 2.0K	&	70	&	 7.9K	&	91	&	80	&	212	&	75	&	0	&	1	&	8	&	1	&	1	&	1	&	1	&	0		\\ \hline
\ti{EAT}	&	 23.2K	&	 325.0K	&	 20.1K	&	 284.8K	&	 7.4K	&	 295.2K	&	 76.9K	&	 44.8K	&	 25.0K	&	 26.9K	&	 4.1K	&	1	&	5	&	3	&	3	&	2	&	3	&	3		\\ \hline
\ti{emailEuAll}	&	 265.0K	&	 419.0K	&	 54.5K	&	 310.0K	&	 1.0K	&	 44.9K	&	 65.1K	&	 22.4K	&	 15.2K	&	 69.7K	&	 49.0K	&	1	&	4	&	7	&	4	&	2	&	6	&	11		\\ \hline
\ti{cit-HepPh}	&	 34.5K	&	 421.5K	&	657	&	 420.2K	&	65	&	 1.3M	&	 5.4K	&	 5.3K	&	232	&	191	&	18	&	1	&	23	&	2	&	2	&	1	&	2	&	1		\\ \hline
\ti{Slashdot}	&	 77.4K	&	 828.2K	&	 359.0K	&	 110.2K	&	89	&	 9.6K	&	 42.8K	&	 16.4K	&	 10.0K	&	 71.1K	&	 401.9K	&	1	&	3	&	6	&	3	&	2	&	3	&	33		\\ \hline
\ti{web-ND}	&	 325.7K	&	 1.5M	&	 379.6K	&	 710.5K	&	 9.5K	&	 499.2K	&	 309.7K	&	 1.2M	&	 40.6K	&	 106.6K	&	 6.8M	&	1	&	15	&	14	&	11	&	2	&	3	&	148		\\ \hline
\ti{amazon}	&	 403.4K	&	 3.4M	&	 944.0K	&	 1.5M	&	45	&	 632.1K	&	 974.9K	&	 627.1K	&	 58.1K	&	 821.5K	&	 872.8K	&	1	&	5	&	6	&	4	&	2	&	4	&	9		\\ \hline
\ti{wiki-Talk}	&	 2.4M	&	 5.0M	&	 361.8K	&	 4.3M	&	 171.9K	&	 227.7K	&	 1.0M	&	 1.6M	&	 1.1M	&	 2.2M	&	 836.5K	&	2	&	12	&	7	&	7	&	6	&	18	&	18		\\ \hline
\ti{soc-pokec}	&	 1.6M	&	 30.6M	&	 8.3M	&	 14.0M	&	 142.7K	&	 5.0M	&	 4.1M	&	 3.9M	&	 2.1M	&	 10.4M	&	 7.0M	&	2	&	25	&	9	&	5	&	3	&	7	&	18		\\ \hline
\ti{liveJournal}	&	 4.8M	&	 68.5M	&	 25.6M	&	 17.2M	&	 202.3K	&	 58.3M	&	 33.9M	&	 46.5M	&	 6.6M	&	 59.7M	&	 80.6M	&	7	&	133	&	98	&	89	&	27	&	65	&	247		\\ \hline
\ti{en-wiki}	&	 4.2M	&	 101.3M	&	 9.4M	&	 82.6M	&	 2.8M	&	 163.3M	&	 61.4M	&	 34.8M	&	 13.8M	&	 22.1M	&	 5.9M	&	7	&	26	&	22	&	24	&	5	&	18	&	29		\\ \hline
\end{tabular}
\hspace{-6ex}
\label{tab:directeds}
\vspace{-2ex}
\end{table*}

\vspace{-2ex}
\section{Experiments}\label{sec:exps}
We evaluate our framework on three types of networks and motifs therein; directed (\cref{sec:expdirected}), signed-directed (\cref{sec:expsigneddirected}), and node-labeled (\cref{sec:expnodelabeled}) networks.
We implement quark decompositions for various motifs in each type and evaluate the resulting subgraphs.
All experiments are performed on a Linux operating system (v. 4.12.14-150.52) running on a machine with Intel(R) Xeon(R) CPU E5-2698 v3 processor at 2.30GHz with 64 GB DDR3 1866 MHz memory. 
Algorithms are implemented in C++ and compiled using gcc 6.1.0 at the -O2 level.
\textbf{The code is available at \url{http://sariyuce.com/quark\_decomposition.tar}}.
For each network type, we discuss the set of motifs used and present the results.
We compare quark decomposition to the state-of-the-art methods and highlight anecdotal examples to stress the contrast between our method and others.
We also present the runtime performance of quark decompositions and other state-of-the-art methods.

\noindent {\bf Baselines.} We consider three baselines in our comparisons.
\begin{compactitem}[\leftmargin=-.5in$\bullet$]
\item {\bf Motif clustering (\MC)~\cite{Benson16}.} Set of higher-order clustering algorithms by Benson et al. (see~\cref{sec:rel} for details). We consider three versions; (1) \MCsingle: Algorithm that gives a single subgraph with near-optimal motif conductance, (2) \MCrecbi: Recursive bisection algorithm that iteratively finds multiple clusters (starting with the optimal) until the cluster size gets too small (less than 10 nodes) or quality degrades too much (conductance goes above 0.5), (3) \MCkmeans: k-means algorithm that is run on the motif adjacency matrix -- number of clusters (k) must be specified for this version.
\item {\bf Takaguchi and Yoshida (\TY)~\cite{Takaguchi16}.} Cycle-truss and flow-truss algorithms (see~\cref{sec:rel} for details).
\item {\bf (r,s) nucleus~\cite{Sariyuce15}.} Nucleus decomposition to find hierarchical dense subgraphs in undirected networks (see~\cref{sec:rel}).
\end{compactitem}

\noindent {\bf Metrics.} We consider three metrics to measure the quality of the subgraphs. We also show anecdotal examples when feasible.
\begin{compactitem}[\leftmargin=-.5in$\bullet$]
\item {\bf Motif conductance.} Edge conductance is adapted for motifs in~\cite{Benson16}. Motif conductance of a subgraph is defined as the ratio of the number of motif instances cut (i.e., motifs in the boundary) to the number of motif instance end points in the subgraph (i.e., nodes participating in the motifs). The lower values are better.
\item {\bf Average motif degree.} Conductance metric is known to have a bias toward giving better results for smaller numbers of clusters~\cite{almeida2011there, Leskovec08}. As an alternative, we consider the average motif degree in each subgraph, as given in~\cref{def:amd}.
In edge-based clustering literature, the densest subgraph of a graph is defined as the one with the largest average degree~\cite{Goldberg84, Tsourakakis15}. Here we adapt this measure for the motif-based subgraphs and simply consider the number of motifs per node. The higher values are better.
\item {\bf Edge density.} We also consider the ratio of edges over all possible in a subgraph ($|E|/{|V| \choose 2}$). We use this metric in~\cref{sec:expnodelabeled} for undirected networks. The higher values are better.
\end{compactitem}

\vspace{-2ex}
\subsection{Directed networks}\label{sec:expdirected}

{\bf Datasets.} We consider several directed networks from various domains in our experiments: Florida Bay food web (\ti{foodweb}), word associations (\ti{EAT}), emails (\ti{email-EuAll}), citations (\ti{cit-HepPh}), online social networks (\ti{slashdot}, \ti{soc-pokec}, \ti{livejournal},\\ \ti{wiki-Talk}), web networks (\ti{web-ND}, \ti{wiki-Talk}), and product co-purchasing relations (\ti{amazon}).
All networks (except \ti{EAT}~\cite{kiss1973associative}) are obtained from SNAP~\cite{snap}.
~\cref{tab:directeds} gives several statistics, including the motif counts and maximum quark numbers.

\begin{figure*}[!t]
\centering
\begin{subfigure}{.25\textwidth}
	\centering
	\includegraphics[width=\linewidth]{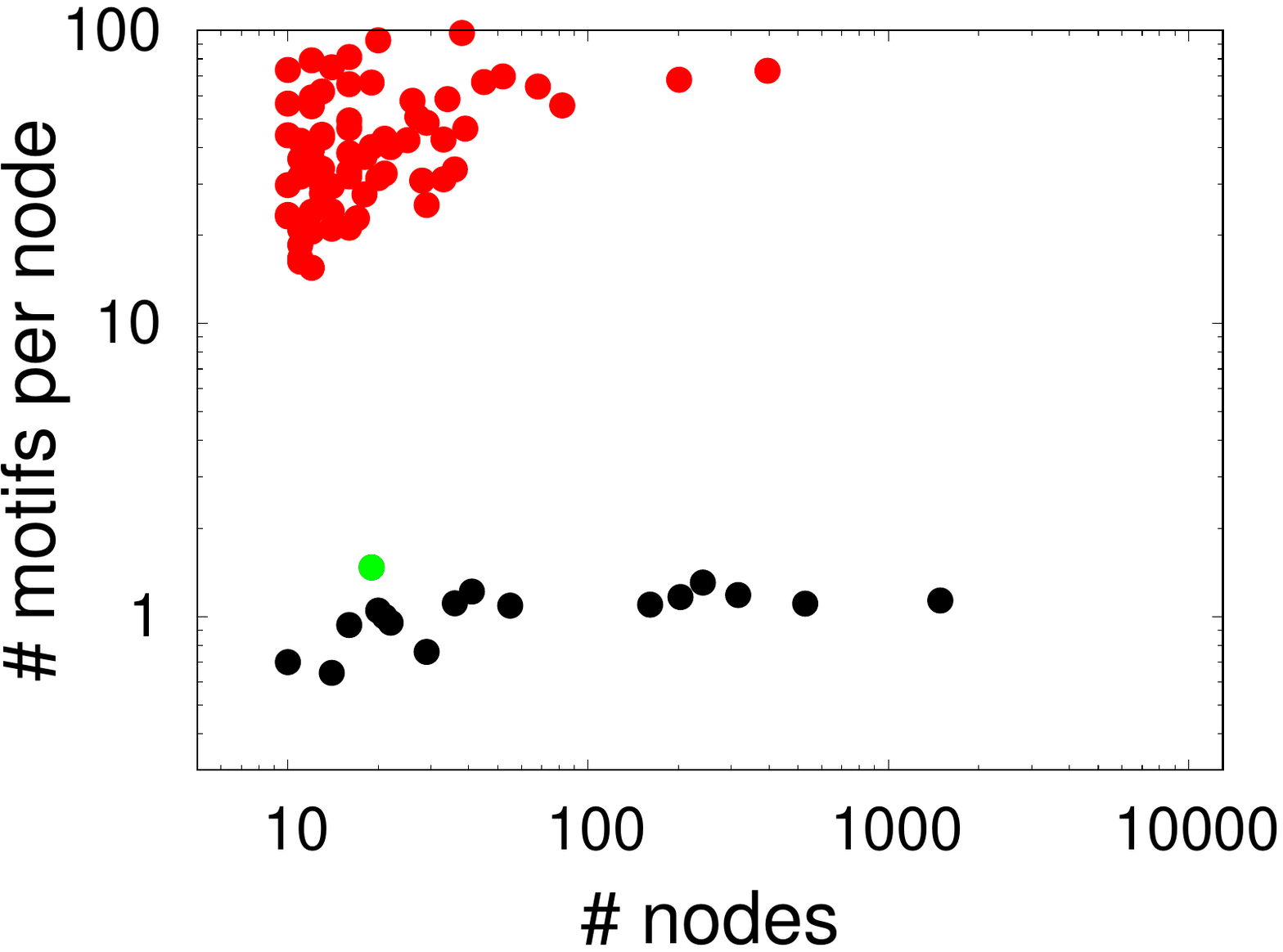}
	\vspace{-5ex}
	
	\includegraphics[width=\linewidth]{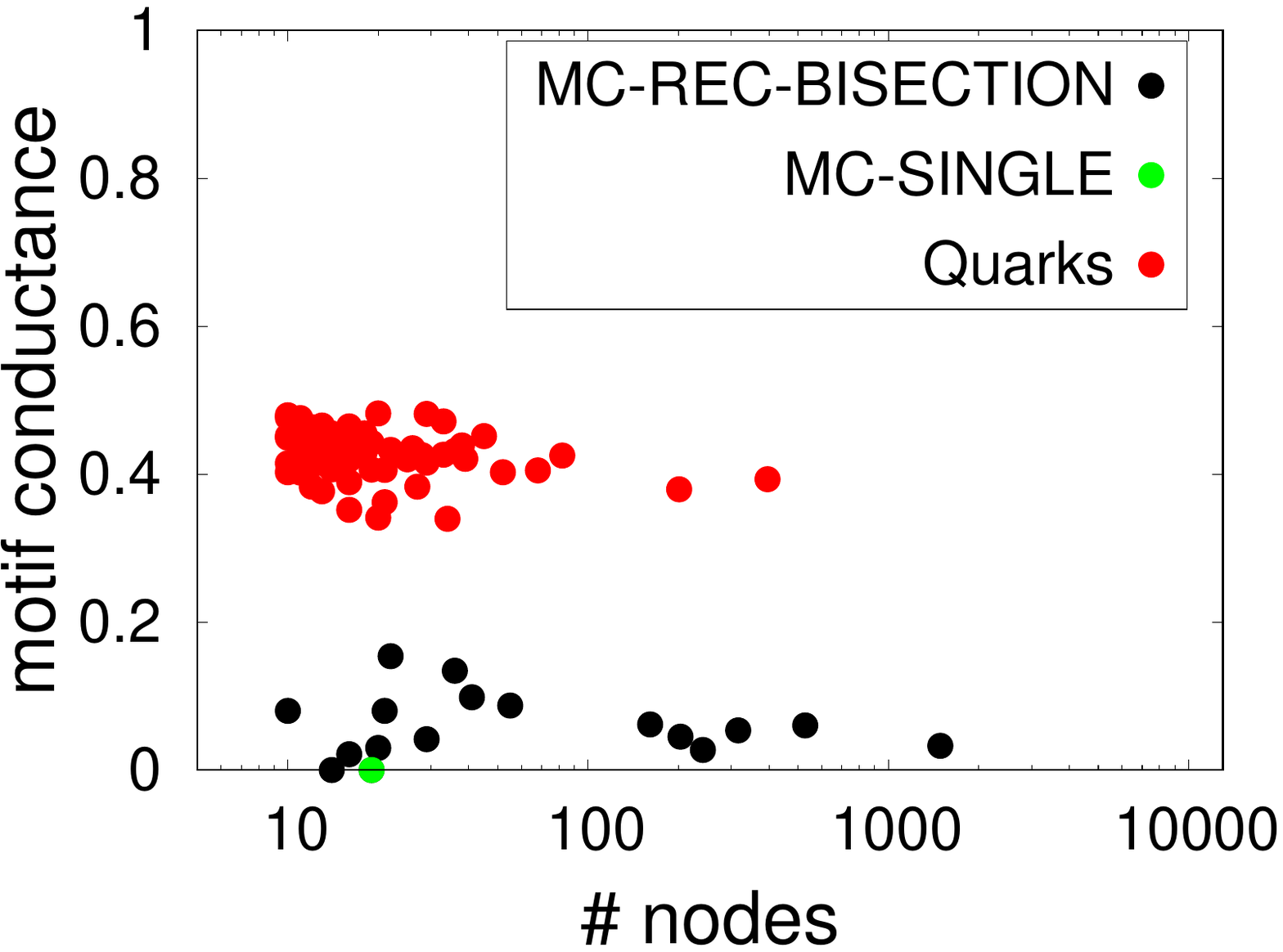}
	\vspace{-4ex}
        \caption{\ita{reciprocal}}
        \label{fig:C1}
\end{subfigure}
\hspace{-2ex}
\begin{subfigure}{.25\textwidth}
	\centering
	\includegraphics[width=\linewidth]{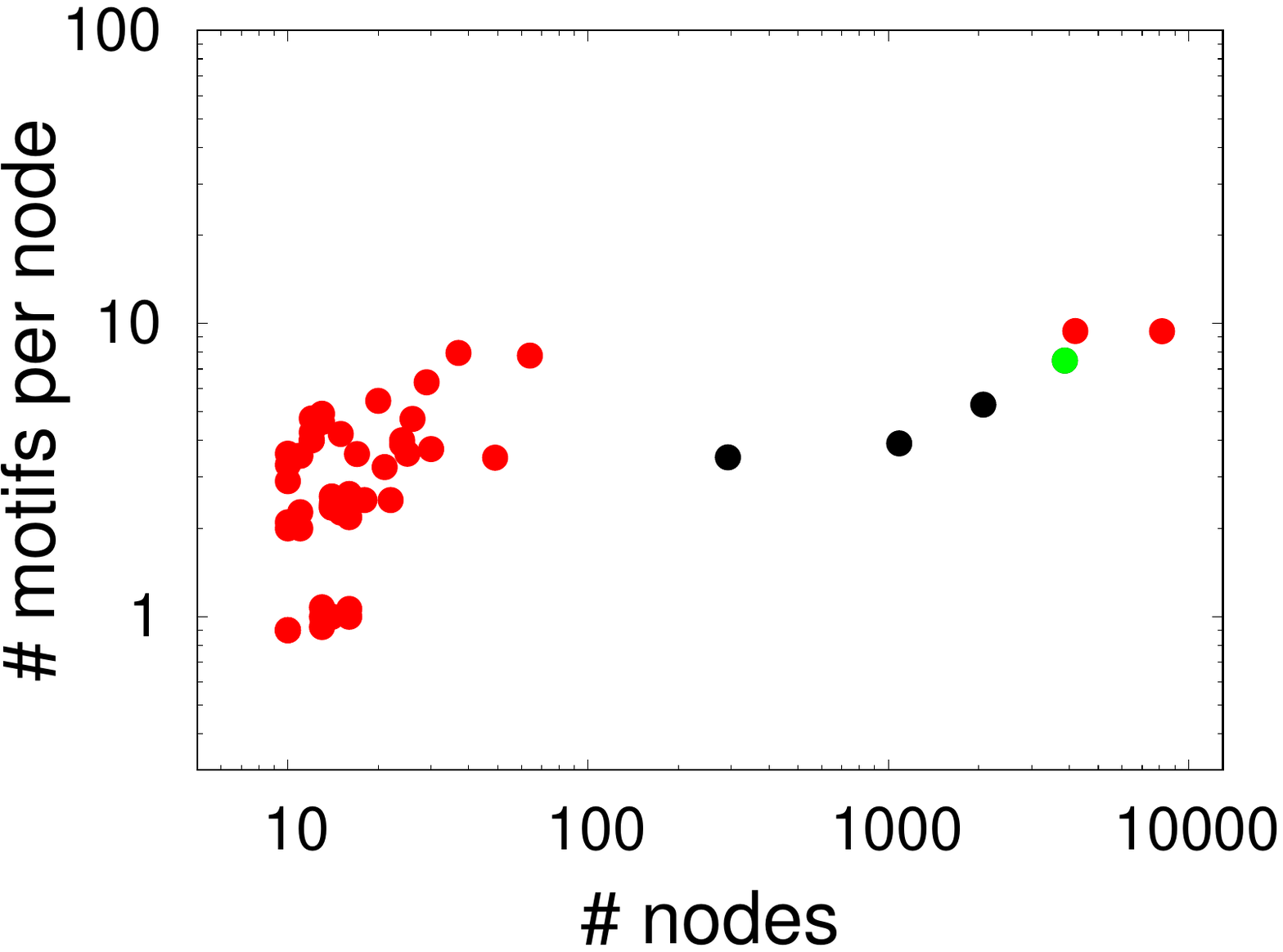}
	\vspace{-5ex}
	
	\includegraphics[width=\linewidth]{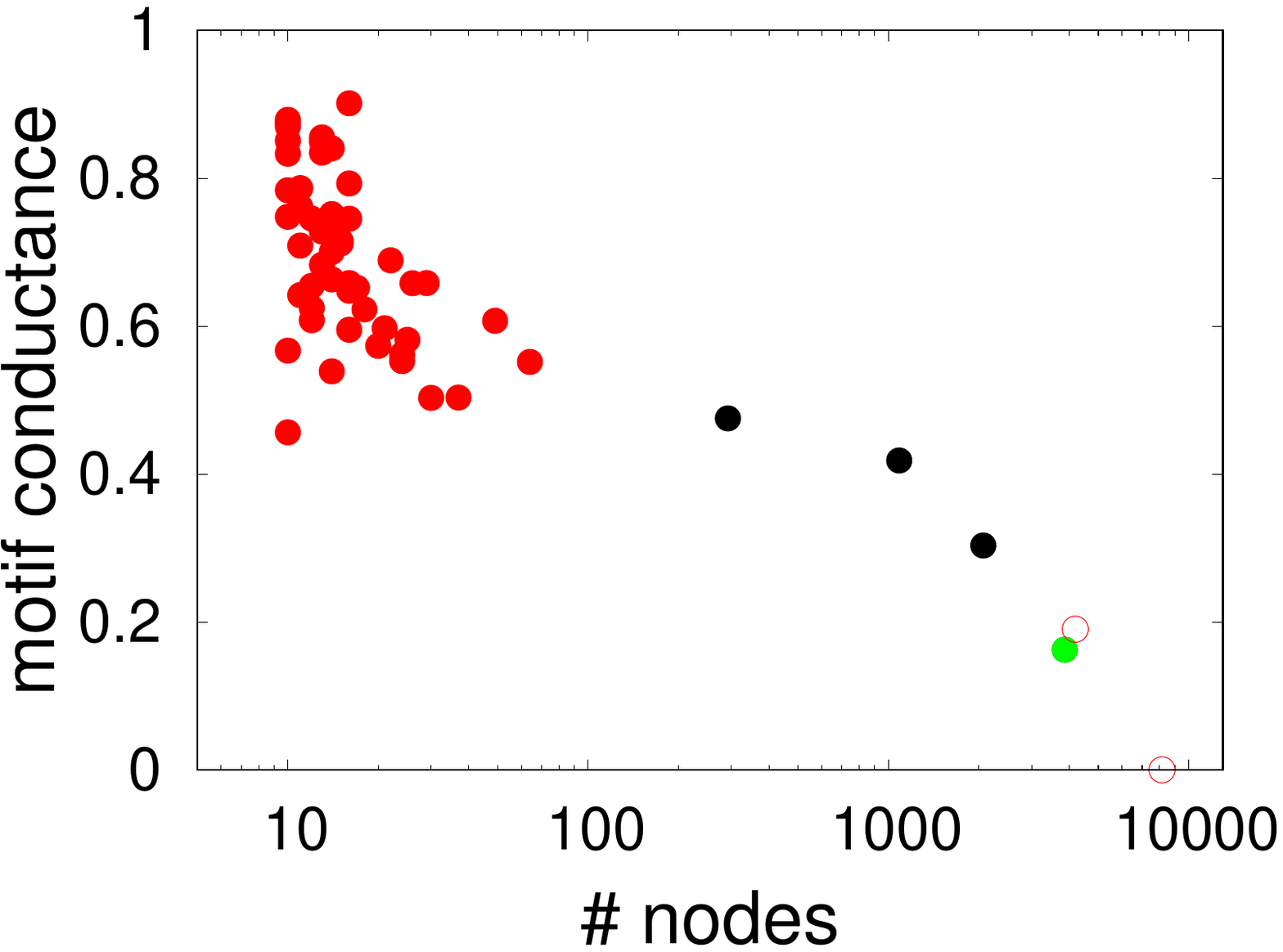}
	\vspace{-4ex}
        \caption{\ita{out+}}
        \label{fig:C4}
\end{subfigure}
\hspace{-2ex}
\begin{subfigure}{.25\textwidth}
	\centering
	\includegraphics[width=\linewidth]{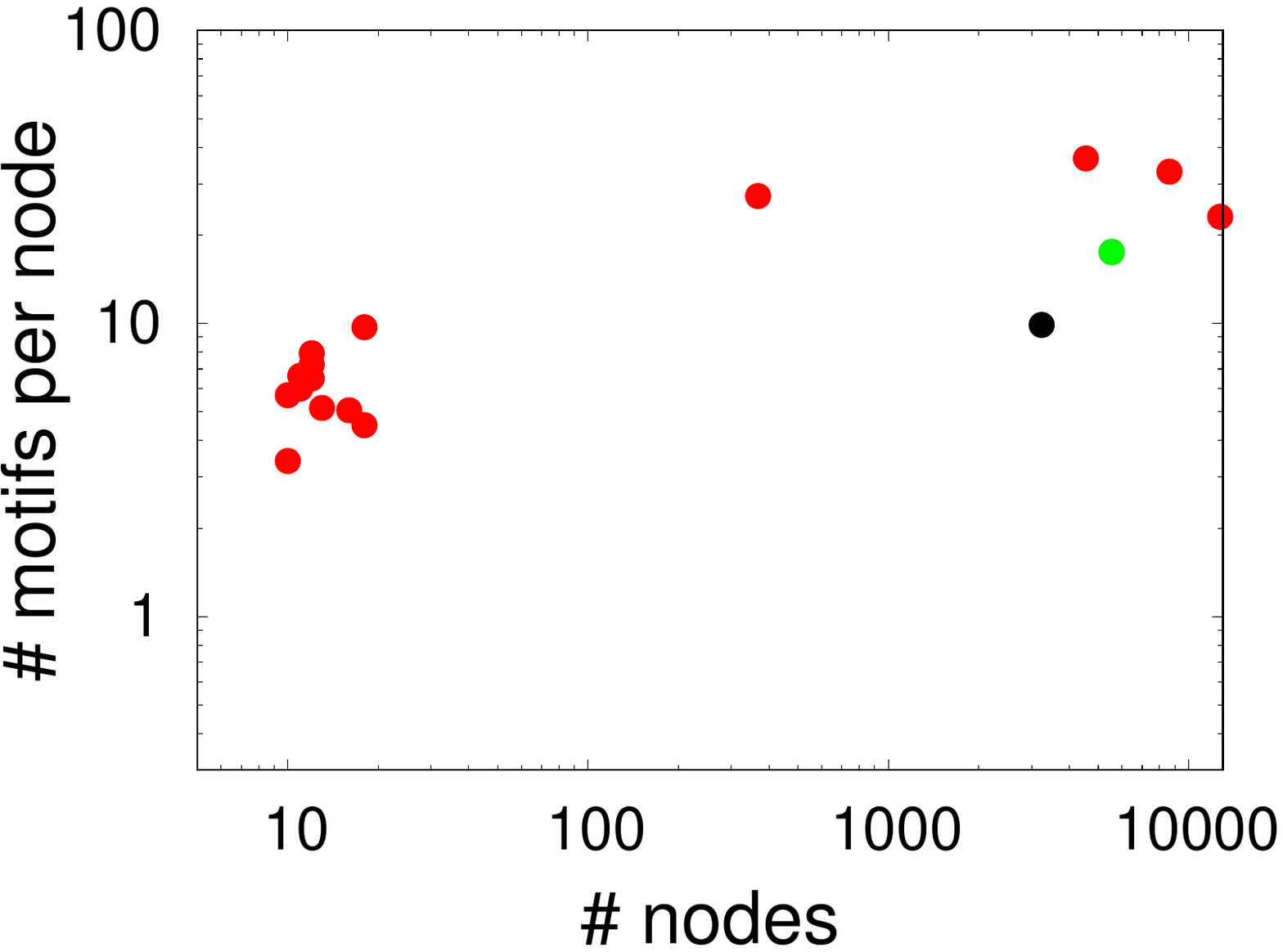}
	\vspace{-5ex}
	
	\includegraphics[width=\linewidth]{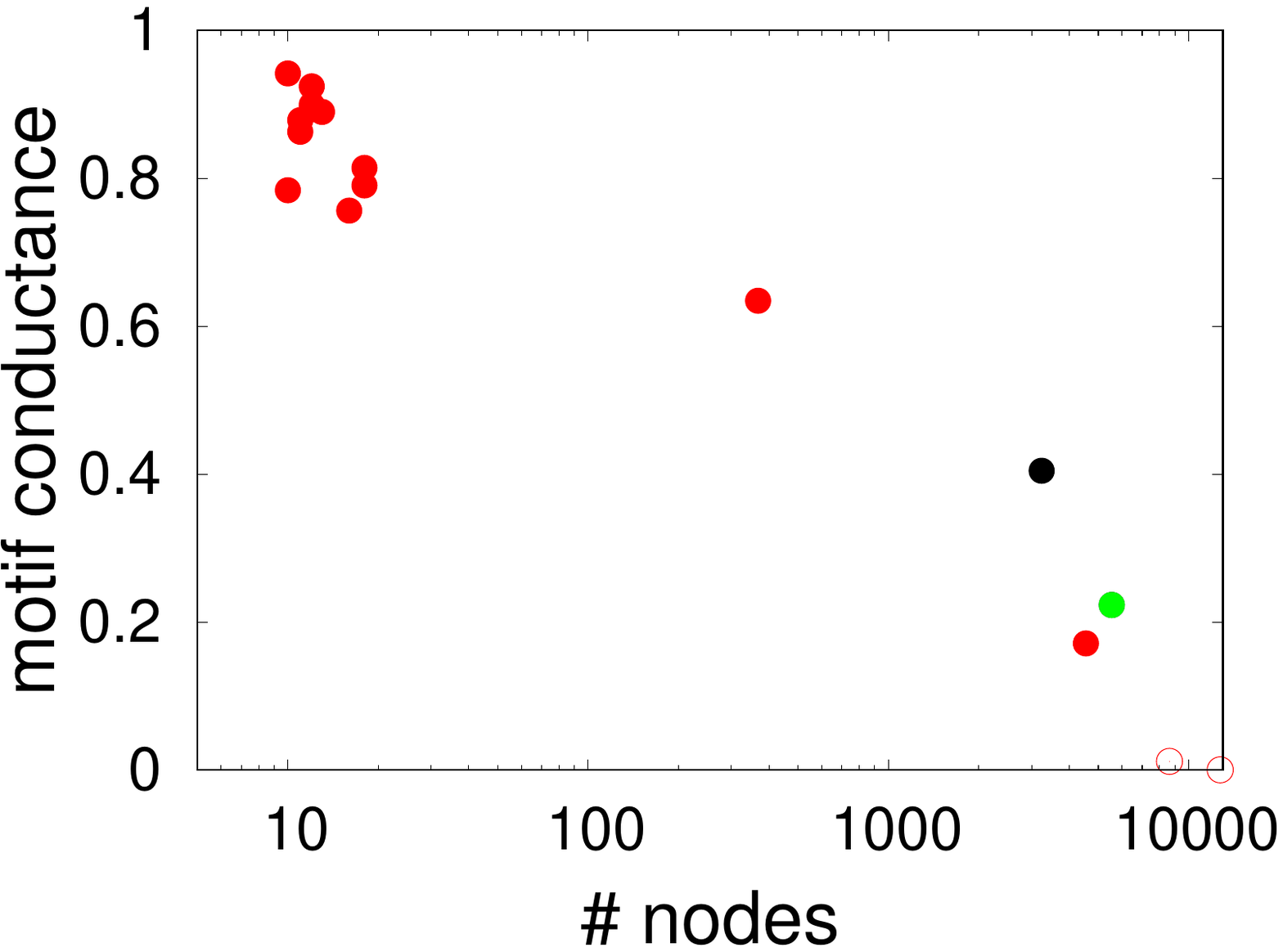}
	\vspace{-4ex}
        \caption{\ita{acyclic}}
        \label{fig:C3}
\end{subfigure}
\hspace{-2ex}
\begin{subfigure}{.25\textwidth}
	\centering
	\includegraphics[width=\linewidth]{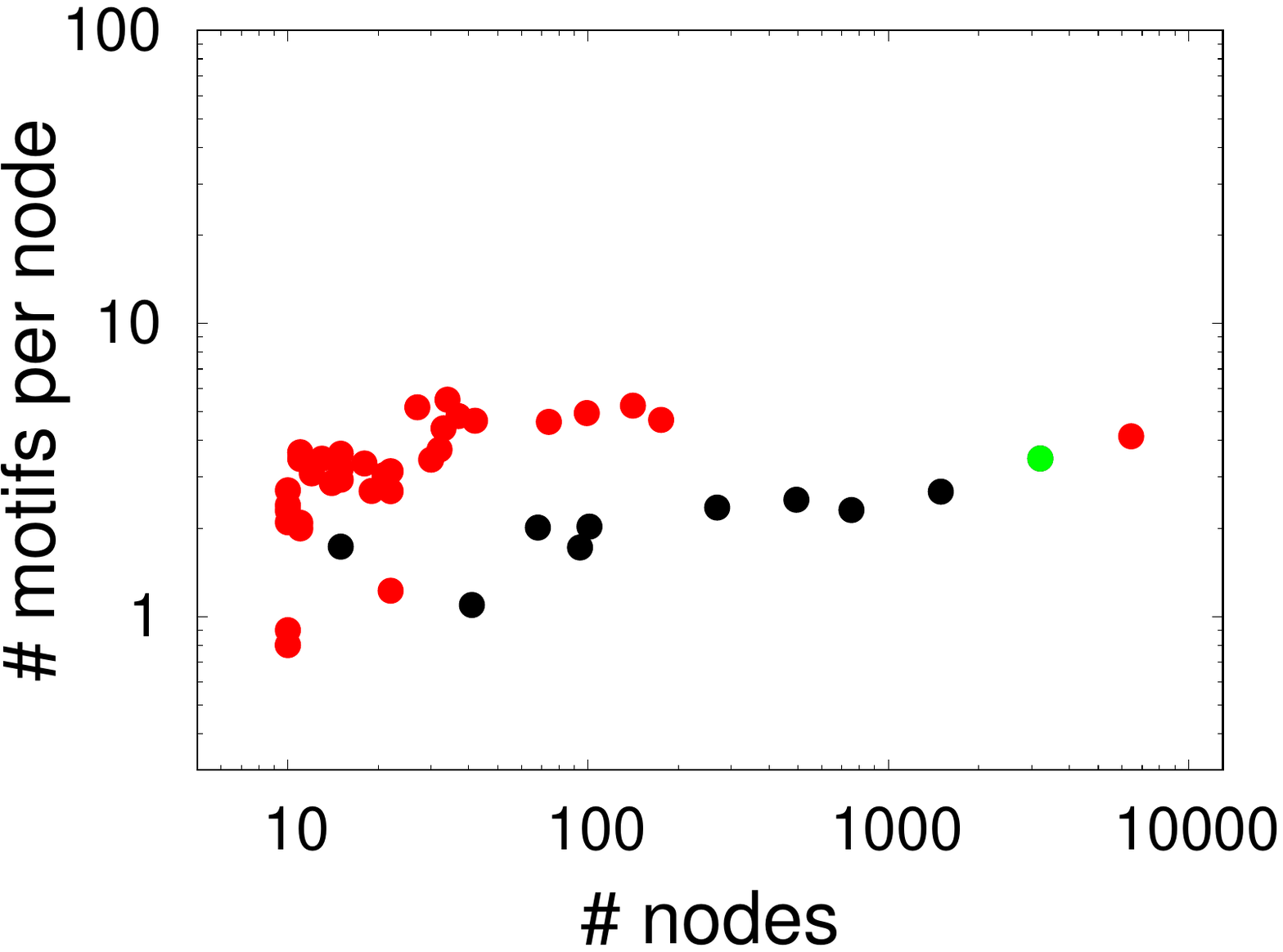}
	\vspace{-5ex}
	
	\includegraphics[width=\linewidth]{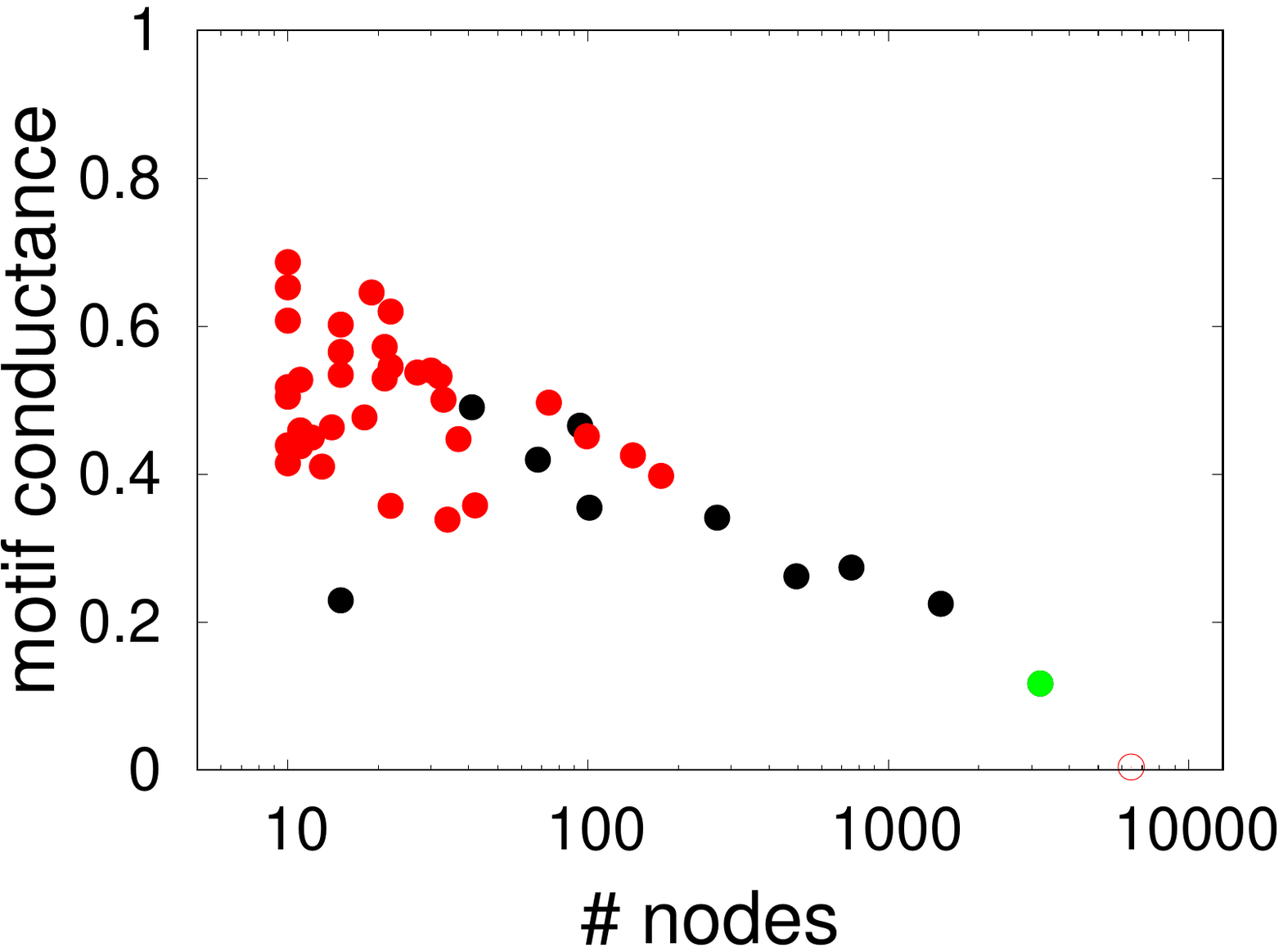}
	\vspace{-4ex}
        \caption{\ita{cycle++}}
        \label{fig:C7}
\end{subfigure}
\vspace{-3ex}
\caption{\small Comparison of the quark decomposition and \MC~\cite{Benson16} in \ti{EAT} network. Top row shows the results for number of motifs per node, i.e., average motif degree, (higher is better) and the bottom row has the motif conductance (lower is better).
We consider \MCsingle, which obtains near-optimal motif conductance, and \MCrecbi, which is applied until the resulting cluster gets too small (less than 10 vertices) or high conductance (more than 0.5).
For quark decomposition, we only show the $k$-quarks with at least 10 nodes. Each subgraph is denoted by a point; the size is shown on the $x$-axis and the metric is given on the $y$-axis. The large quarks for which the conductance computation requires the rest of the graph (since the size is more than the half) are denoted by red circles for completeness (conductances for those are not real).
}
\label{fig:EATcomp}
\vspace{-3ex}
\end{figure*}

\noindent {\bf Motifs.} We instantiate the quark decomposition for directed networks by considering the edge and triangle motifs (\cref{fig:dirmotifs}), corresponding to the $M$ and $N$ in~\cref{def:quark}, respectively.
Note that considering edge as motif $M$ is more advantageous than node. Since the edges are assigned quark numbers, $k$-quarks can overlap with each other.
Also, role confusion does not happen for \ita{out+} and \ita{in+} as explained in~\cref{sec:roleconf}.
We also incorporate the reciprocity by considering the unidirectional and bidirectional edges separately, rather than treating each bidirectional edge as two unidirectional edges.
This is because directed networks often have a significant percentage of bidirectional edges (also observed in~\cref{tab:directeds}) and those need to be treated differently, as discussed in ~\cite{Newman02, Sesh16}.
{\it We use the vanilla edge as the motif $M$ and each of the seven directed triangles (\cref{fig:dirmotifs}) as the motif $N$.}

We first discuss the motif counts and quarks for directed triangles. Then we compare quarks with baselines using the {\ti EAT} and {\ti foodweb} networks. We finish by comparing the runtimes.

\subsubsection{\bf Motif counts and subgraphs}

\cref{tab:directeds} lists the motif counts and maximum quark numbers for each motif.
\ita{cycle} is often the least common motif.
Networks with significant fraction of bidirectional edges tend to contain \ita{reciprocal} motifs the most.
Note that the most frequent motif does not always yield the highest quark number.
In particular, \ita{reciprocals} are often concentrated in a small region, thus yield highly-dense subgraphs.
For instance, \ti{en-wiki} has 5.9M \ita{reciprocal} and 163.3M \ita{acyclic} motifs but the maximum quark numbers for those motifs are 29 and 26, respectively. 
This is because the maximum $k$-quark has 465 edges in the \ita{reciprocal} case but has 9461 edges in \ita{acyclic}.
We also observe that the number of quarks are independent of the motif counts.
For example, the most prevalent motif in \ti{EAT} network is, by far, \ita{acyclic}.
However, \ita{acyclic} yields only 447 $k$-quarks while the \ita{cycle+} locates 2826 subgraphs.
Overall, the abundance of a particular motif does not imply the existence of dense subgraphs containing that motif.

\subsubsection{\bf Comparison with previous methods}

We compare quark decomposition against the three baselines listed above. We use \ti{EAT} network, a collection of word association norms where the nodes are English words and an edge $(u, v)$ implies that human subjects consider the word $v$ when they are shown the word $u$ as stimulus.

We first compare the quality of subgraphs given by the quark decomposition and \MC algorithms~\cite{Benson16} using motif conductance and average motif degree metrics. For quark decomposition, we only consider the $k$-quarks with at least 10 nodes.
For \MC algorithms, we consider \MCsingle and \MCrecbi.
Note that we aim to compare the `best' clusters reported by the two algorithms and do not intend any comparison with respect to the number of clusters reported.
\cref{fig:EATcomp} gives the comparison for \ita{reciprocal}, \ita{out+}, \ita{acyclic}, and \ita{cycle++} motifs (other motifs show similar results and omitted).
For each subgraph, we report the size as well as the quality metric; number of motifs per node in the top row and the motif conductance in the bottom row.
\MC variants often give large subgraphs, always with very low conductance, as expected. \MCsingle (optimal cluster) has more than 1000 nodes for \ita{out+}, \ita{in+}, \ita{acyclic}, and \ita{cycle++}.
For \ita{cycle+}, however, it has only five nodes.
Quarks are often small, most in the range of 10-100 nodes, and have higher conductance scores. For \ita{out+} and \ita{cycle++}, some quarks yield comparable conductance scores with \MCrecbi.
We also observe that a quark for \ita{acylic} has a better conductance than \MCsingle (which is the near-optimal as shown in~\cite{Benson16}).
Regarding the average motif degrees (top row), quarks perform significantly better in all the motifs (note that y-axis is in log-scale). 
By definition of the quark (the connectivity constraint in particular), if the size is $n$, the number of motifs is at least $n-2$  (for $k=1$, a single motif is a valid subgraph and can be extended by a new node that creates a new motif, keeping the motif count $n-2$). This ensures a lower bound, $\frac{n-2}{n}$ (close to 1), for average motif degree in quarks.
In general, quarks tend to be smaller in size when compared to \MC results, have consistently higher average motif degrees, and comparable conductance scores for some motifs.
Overall, the top-down partitioning approach in the motif clustering is likely to result in larger subgraphs in varying quality whereas the bottom-up computation in quark decomposition yields smaller subgraphs with larger average motif degrees.

Next we compare quark decomposition with \TY~\cite{Takaguchi16}.
For \ti{EAT} network, \TY reports maximum cycle-truss number of 3 and maximum flow-truss number of 10. For each maximum truss subgraph, we checked the quarks that are the most similar.
\cref{fig:tycomp} presents the results for cycle-truss, flow-truss, and their corresponding quarks with size and average motif degree information. For cycle- and flow-truss, we calculate the induced cycle and acyclic motif degrees (i.e., bidirectional edges are not included).
Next to each quark, we denote the size of its intersection with the truss.
Various types of quarks are able to obtain almost all the nodes in those trusses.
70 of 77 nodes in cycle-truss are obtained with 15 quarks and all of the 45 nodes in flow-truss are given in 15 other quarks.
This verifies the artificial over-representation of cycle- and flow-trusses due to the non-induced nature.
Overall, treating the bidirectional edges as atomic units enables finding diverse subgraphs while correctly capturing the semantics of pairwise relationships.

\begin{figure}[!b]
\centering
\begin{subfigure}{.5\linewidth}
	\centering
	\includegraphics[width=\linewidth]{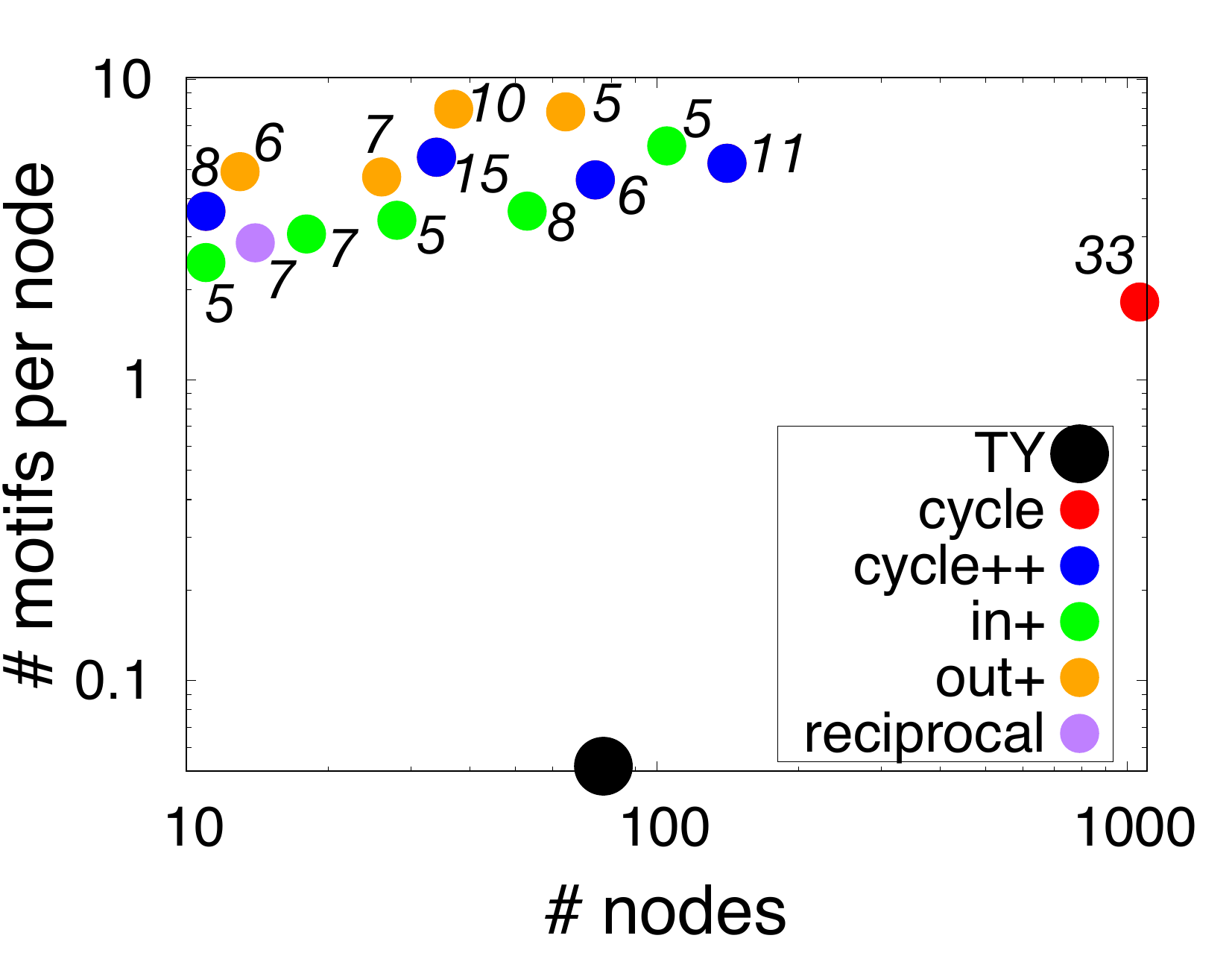}
	\vspace{-4ex}
        \caption{cycle-truss vs. quarks}
        \label{fig:ty1}
\end{subfigure}
\hspace{-2ex}
\begin{subfigure}{.5\linewidth}
	\centering
	\includegraphics[width=\linewidth]{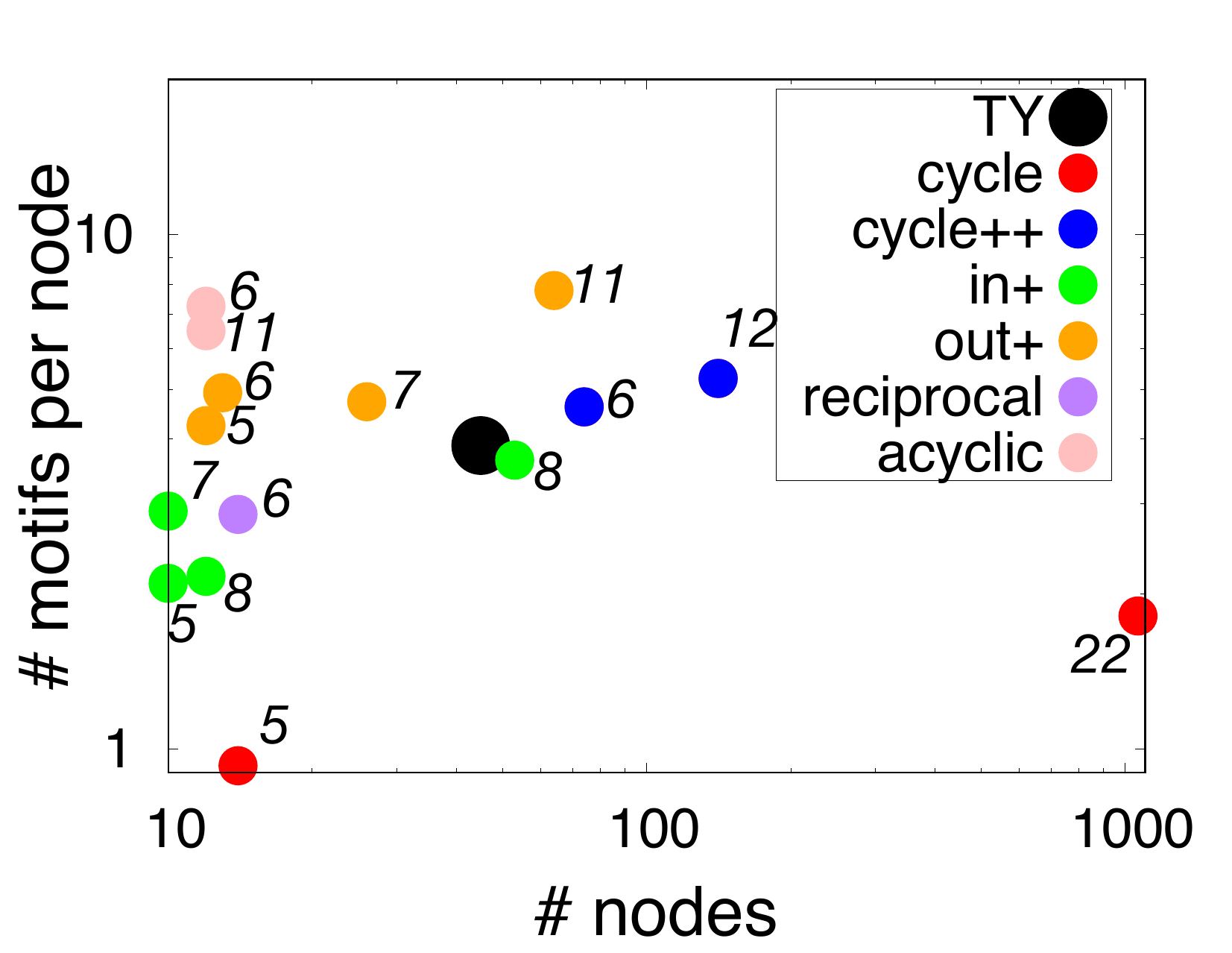}
	\vspace{-4ex}
        \caption{flow-truss vs. quarks}
        \label{fig:ty2}
\end{subfigure}
\vspace{-3ex}
\caption{\small Comparison of cycle- and flow-truss to quarks in \ti{EAT} network. For each quark, size of its intersection with the truss is shown.}
\label{fig:tycomp}
\end{figure}

\begin{figure}[!t]
\centering
\includegraphics[width=0.8\linewidth]{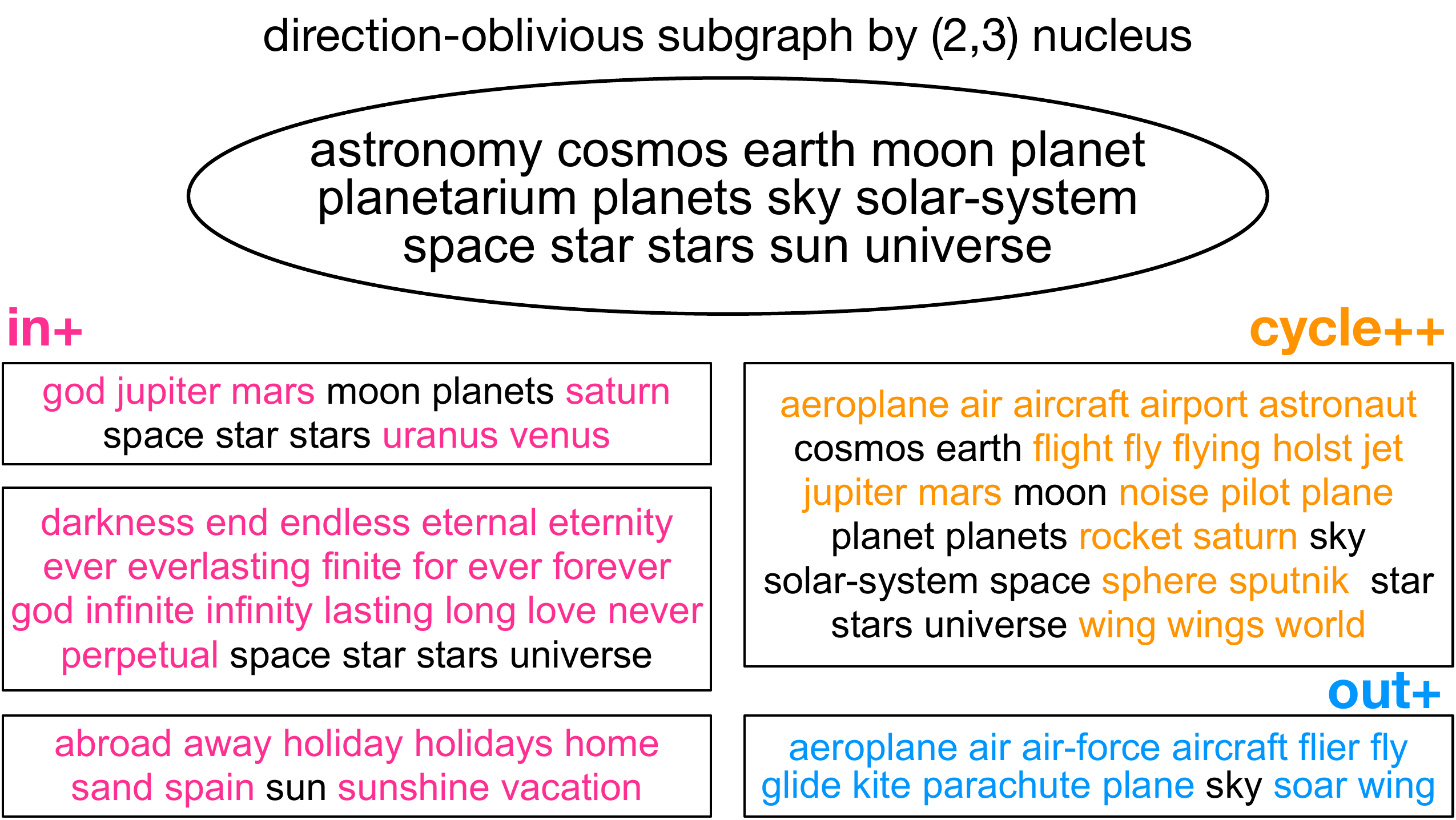}
\vspace{-2ex}
\caption{\small Comparison of a direction-oblivious $(2,3)$ nucleus and various quarks in \ti{EAT} network. The common words in quarks and the nucleus is shown in black. Quarks by different motifs capture different contexts for those words. \ita{in+} provides multiple contexts for \textit{space, stars} thanks to the fact that quarks are overlapping.}
\label{fig:EATexample}
\end{figure}

Lastly, we compare quarks with $(2,3)$ nucleus decomposition~\cite{Sariyuce15}, which ignores edge directions.
We observe that incorporating the edge directions results in more diverse subgraphs.
The number of subgraphs (of any quality) obtained by each quark decomposition is significantly larger than what nucleus decomposition yields.
As an anecdotal example, we show a subgraph found by the nucleus decomposition in~\cref{fig:EATexample}.
The direction-oblivious subgraph contains words related to astronomy and space.
Quark decompositions capture several diverse contexts related to those words.
Thanks to the overlapping quarks, \ita{in+} finds two subgraphs that contain \textit{space} and \textit{stars}: one in astronomy theme (\textit{uranus, venus, ...}) and another in the religious context (\textit{god, eternity, ...}).
It also finds another subgraph in the vacation theme (\textit{sun, holiday, sand, ...}).
\ita{Out+} yields a subgraph in the air-flight context (\textit{sky, aircraft, wing, ...}).
We also recognize that multiple meanings of the homonym words are reflected in various $k$-quarks.
For instance, \textit{lie} is reported in two subgraphs by \ita{cycle++}: one is about incorrectness (\textit{falsehood, untruth, ...}) and the other is about staying at rest in the horizontal position (\textit{couch, rest, ...}).
Overall, quark decompositions by various motifs can locate diverse contexts for a given word thanks to the motif-aware approach and overlapping nature of quarks.

\subsubsection{\bf Analysis of Florida Bay food web}

Here we analyze the structure of Florida Bay food web network ({\ti foodweb}) where the nodes are the compartments (i.e., organisms, species) and the edges are the directed carbon exchanges (i.e., $u\rightarrow v$ if $v$ eats $u$).
Benson et al. showed that high-quality clusters (i.e., with low conductance) by \MC only exist for \ita{out+}, which implies that the organization of compartments is better described with \ita{out+} (as opposed to the common belief that \ita{acyclic} is the key motif)~\cite{Benson16}. They also show that the 4 clusters by \MCkmeans for \ita{out+} reflect the ground-truth subgroup classifications better than the state-of-the-art clustering algorithms such as spectral edge clustering (with k-means and recursive bisection)~\cite{von2007tutorial}, InfoMap~\cite{rosvall2008maps}, and Louvain method~\cite{blondel2008fast}.

We first compare the quarks for \ita{out+} with \MCkmeans.
We find 7 quarks for \ita{out+}, so we consider \MCkmeans with 7 clusters as well as with 4 clusters.
The nodes that appear in multiple quarks are only considered to be part of their largest quark (other choices give similar results).
\cref{tab:foodweb-comparison} presents the results for two ground-truth classifications given in~\cite{ulanowicz1998network, Benson16} by four metrics: Adjusted Rand Index (ARI), F1 score, Normalized Mutual Information (NMI), Purity~\cite{manning2008introduction}.
{\bf Quarks clearly outperforms \MCkmeans variants in both classifications by all the metrics.}
One particular difference is that \MCkmeans considers some macroinvertebrates like predatory crabs among the benthic predators of eels and toadfish whereas quark decomposition finds all macroinvertebrates in the same subgraph.
We believe the main reason is that \MCkmeans considers motif counts from the node-perspective while quark decomposition is based on the edges and their motif counts.

\begin{table}[!t]
\centering
\small
\caption{\small \ti{foodweb} classification results. Best in each row is in bold.}
\vspace{-3ex}
\begin{tabular}{|c|c|c|c|c|} \hline
\multirow{2}{*}{\it{out+}} &	\multirow{2}{*}{Metric}	&	Quarks	&	\MCkmeans	&	\MCkmeans	\\
	&	 	&	(7 subgraphs)	&	w/ 4 clusters	&	w/ 7 clusters	\\ \hline
\parbox[t]{2mm}{\multirow{4}{*}{\rotatebox[origin=c]{90}{Class 1}}} 	&	ARI	&	{\bf 0.3627}	&	0.3005	&	0.1485	\\
	&	F1	&	{\bf 0.4869}	&	0.4574	&	0.3794	\\
	&	NMI	&	{\bf 0.5415}	&	0.5040	&	0.4843	\\
	&	Purity	&	{\bf 0.5968}	&	0.5645	&	0.5161	\\ \hline
\parbox[t]{2mm}{\multirow{4}{*}{\rotatebox[origin=c]{90}{Class 2}}} 	&	ARI	&	{\bf 0.3816}	&	0.3265	&	0.1871	\\
	&	F1	&	{\bf 0.5675}	&	0.5380	&	0.4601	\\
	&	NMI	&	{\bf 0.5206}	&	0.4822	&	0.4309	\\
	&	Purity	&	{\bf 0.6452}	&	0.6129	&	0.5645	\\ \hline\end{tabular}
\label{tab:foodweb-comparison}
\end{table}

\begin{table}[!b]
\centering
\small
\caption{\small Runtimes for large directed networks (sec). Best result in each motif column is shown in bold.}
\vspace{-3ex}
\renewcommand{\tabcolsep}{1.4pt}
\begin{tabular}{|l|r:r|r:r|r:r|r:r|r:r|r:r|} \hline
 &  \multicolumn{2}{c|}{\ita{cycle}} & \multicolumn{2}{c|}\ita{{acyclic}} & \multicolumn{2}{c|}{\ita{out+}} & \multicolumn{2}{c|}{\ita{in+}} & \multicolumn{2}{c|}{\ita{cycle+}} & \multicolumn{2}{c|}{\ita{cycle++}}  \\ 
 &  \multicolumn{1}{c:}{Q    } &  \multicolumn{1}{c|}{M    }  &  \multicolumn{1}{c:}{Q    } &  \multicolumn{1}{c|}{M    } &  \multicolumn{1}{c:}{Q    } &  \multicolumn{1}{c|}{M    } &  \multicolumn{1}{c:}{Q    } &  \multicolumn{1}{c|}{M    } &  \multicolumn{1}{c:}{Q    } &  \multicolumn{1}{c|}{M    } &  \multicolumn{1}{c:}{Q    } &  \multicolumn{1}{c|}{M    } \\ \hline
 \ti{web-ND}	&	{\bf 0.34}	&	3.31	&	{\bf 4.26}	&	16.8	&	{\bf 0.62}	&	6.3	&	{\bf 2.11}	&	8.54	&	{\bf 0.53}	&	10.01	&	{\bf 0.78}	&	9.86			\\
\ti{amzn}	&	{\bf 0.74}	&	3.54	&	{\bf 3.29}	&	79	&	{\bf 2.25}	&	132	&	{\bf 1.92}	&	105	&	{\bf 1.18}	&	5.29	&	{\bf 3.23}	&	107		\\
\ti{wiki}	&	28.9	&	{\bf 14.0}	&	112	&	{\bf 18.2}	&	{\bf 10.9}	&	16.4	&	21.1	&	{\bf 17.7}	&	20.5	&	{\bf 20.2}	&	47.8	&	{\bf 16.8}		\\
\ti{soc-p}	&	{\bf 23.6}	&	79	&	{\bf 66.9}	&	99	&	{\bf 37.0}	&	119	&	{\bf 34.2}	&	139	&	{\bf 48.9}	&	129	&	{\bf 98.1}	&	128		\\
\ti{live-j}	&	{\bf 37.4}	&	200	&	{\bf 180}	&	943	&	{\bf 118}	&	1135	&	{\bf 126}	&	1438	&	{\bf 112}	&	828	&	{\bf 289}	&	2248	\\
\ti{en-w}	&	900	&	{\bf 501}	&	7746	&	{\bf 864}	&	1511	&	{\bf 799}	&	1709	&	{\bf 677}	&	{\bf 398}	&	724	&	2223	&	{\bf 677}		\\ \hline
\end{tabular}
\label{tab:dirtimes}
\end{table}

We also consider \ita{acyclic} and use role-aware quark decomposition (Algorithm~\ref{alg:RAquark}) to determine the roles of the compartments in the resulting quarks.
For an acyclic formed by $u$$\rightarrow$$v$, $u$$\rightarrow$$w$, and $v$$\rightarrow$$w$, we define $u$ as the {\it prey} orbit, $v$ as the {\it balancer} orbit, and $w$ as the {\it predator} orbit. 
The maximum quark obtained by {\sc QuarkDec} (Algorithm~\ref{alg:quark}) and {\sc RA-QuarkDec} (Algorithm~\ref{alg:RAquark}) is the same and contains 48 compartments.
{\sc RA-QuarkDec} assigns three quark numbers for each edge, corresponding to the three edge orbits in \ita{acyclic} (each is a union of its nodes' orbits).
We determine the role profile for each node based on the quark numbers of its outgoing and incoming edges.
The compartments that are dominantly {\it predator} are the birds including predatory ducks, big herons \& egrets, greeb, and more.
Dominantly {\it prey} compartments include clown goby, four types of zooplankton microfauna, and seven macroinvertebrates (including pink and herbivorous shrimps).
Lastly, the ones that have the {\it balancer} role are the fishes, such as (bay) anchovy, sardines, and mojarra.
Note that it is not possible to understand the roles of the compartments by {\sc QuarkDec} since each edge has a single quark number. For instance, the average quark numbers of incoming and outgoing edges for predatory ducks and code goby are very close, which tells nothing about their roles. 

\subsubsection{\textbf{Runtime performance.}}
We measure the runtime for quark decomposition and \MCsingle on all directed networks,~\cref{tab:dirtimes} lists the results for large networks.
\MCsingle (denoted M) gives only one near-optimal cluster and quark decomposition (denoted Q) finds all the quark numbers.
Quark decomposition is faster than \MCsingle for all motifs in \ti{web-ND}, \ti{amazon}, \ti{soc-pokec}, and \ti{liveJournal}.
For some configurations, such as \ita{in+} in \ti{liveJournal}, we observe up to 10x speedup.
For \ti{en-wiki}, however, \MCsingle is faster for all motifs. Note that \MCsingle runtime includes motif adjacency construction and spectral clustering to find one cluster. In order to find more clusters, spectral clustering needs to be run again. However, the spectral clustering takes 36\% of the total time for \ti{en-wiki} (on avg.), hence obtaining 10 clusters will increase the runtime by 4x. 
All in all, although \MCsingle finds only one cluster, quark decomposition is faster for most networks and motifs, and a better choice especially when multiple subgraphs are targeted.

\subsection{Signed-directed networks}\label{sec:expsigneddirected}

{\bf Datasets.} We use signed-directed networks that have categorical labels on edges to denote one-sided positive/negative relationships (no bidirectional edges).
We have two Reddit hyperlink networks that have directed connections among the subreddits.
\ti{reddit-body} and \ti{reddit-title} consider the positive and negative interactions among users who belong to different subreddits~\cite{Kumar18}.
We considered the last interactions in the datasets.
We also consider \ti{epinions}, a who-trust-whom social network~\cite{Guha04},
and \ti{slashdot} which contains the self-tagged friend/foe relationships~\cite{Kunegis09}.
All are obtained from SNAP~\cite{snap}.
\cref{tab:signeds} gives the number of positive and negative edges, motif counts, and maximum quark numbers.

\noindent {\bf Motifs.} We use edge and triangle motifs, corresponding to the $M$ and $N$ in~\cref{def:quark}, respectively.
This also ensures that quarks can overlap with each other.
There is no bidirectional edge and there are twelve possible triangle motifs in total, as shown in~\cref{fig:signedmotifs}; four \ita{cycle} motifs since there is a single orbit and eight \ita{acyclic} motifs where each \texttt{++-} and \texttt{+--} appears in three different ways.
{\it We use the vanilla edge as $M$ and each of the twelve triangles (\cref{fig:signedmotifs}) as $N$.}

\begin{figure}[!t]
\centering
\hspace{-1ex}
\includegraphics[width=1\linewidth]{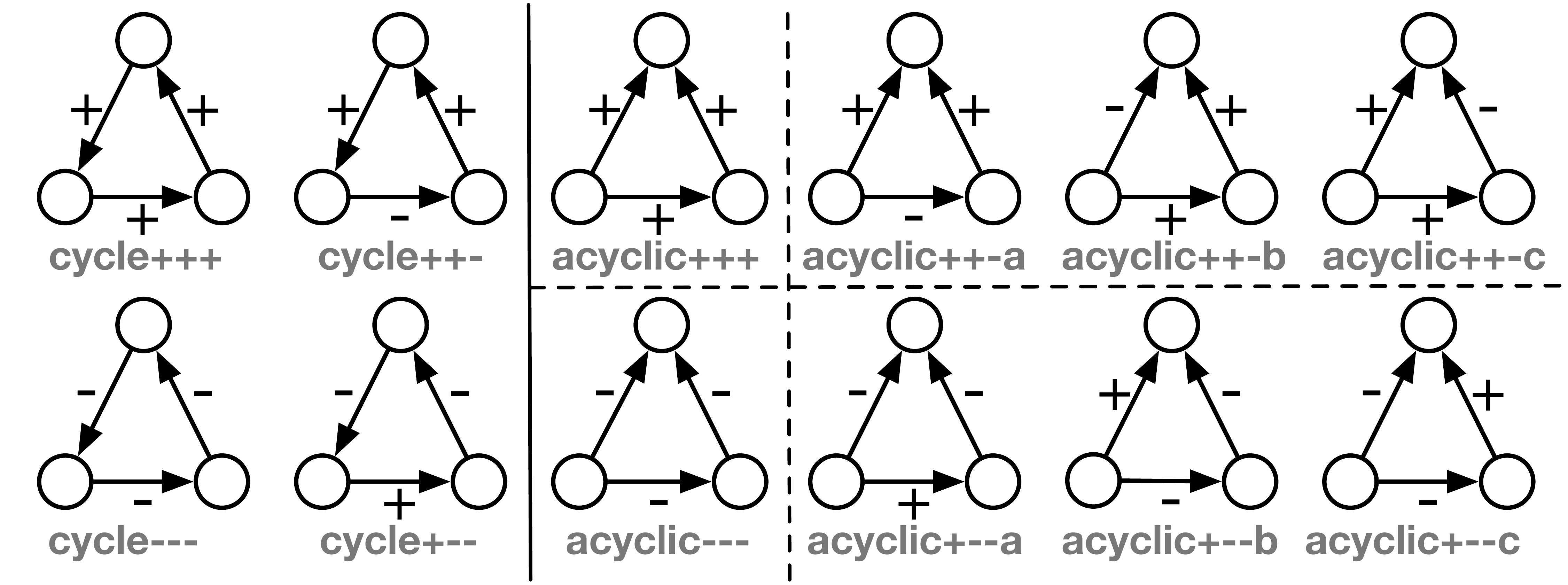}
\hspace{-1ex}
\vspace{-3ex}
\caption{\small Signed directed triangle motifs.}
\label{fig:signedmotifs}
\end{figure}

\begin{figure}[!b]
\centering
\begin{subfigure}{.5\linewidth}
	\centering
	\includegraphics[width=\linewidth]{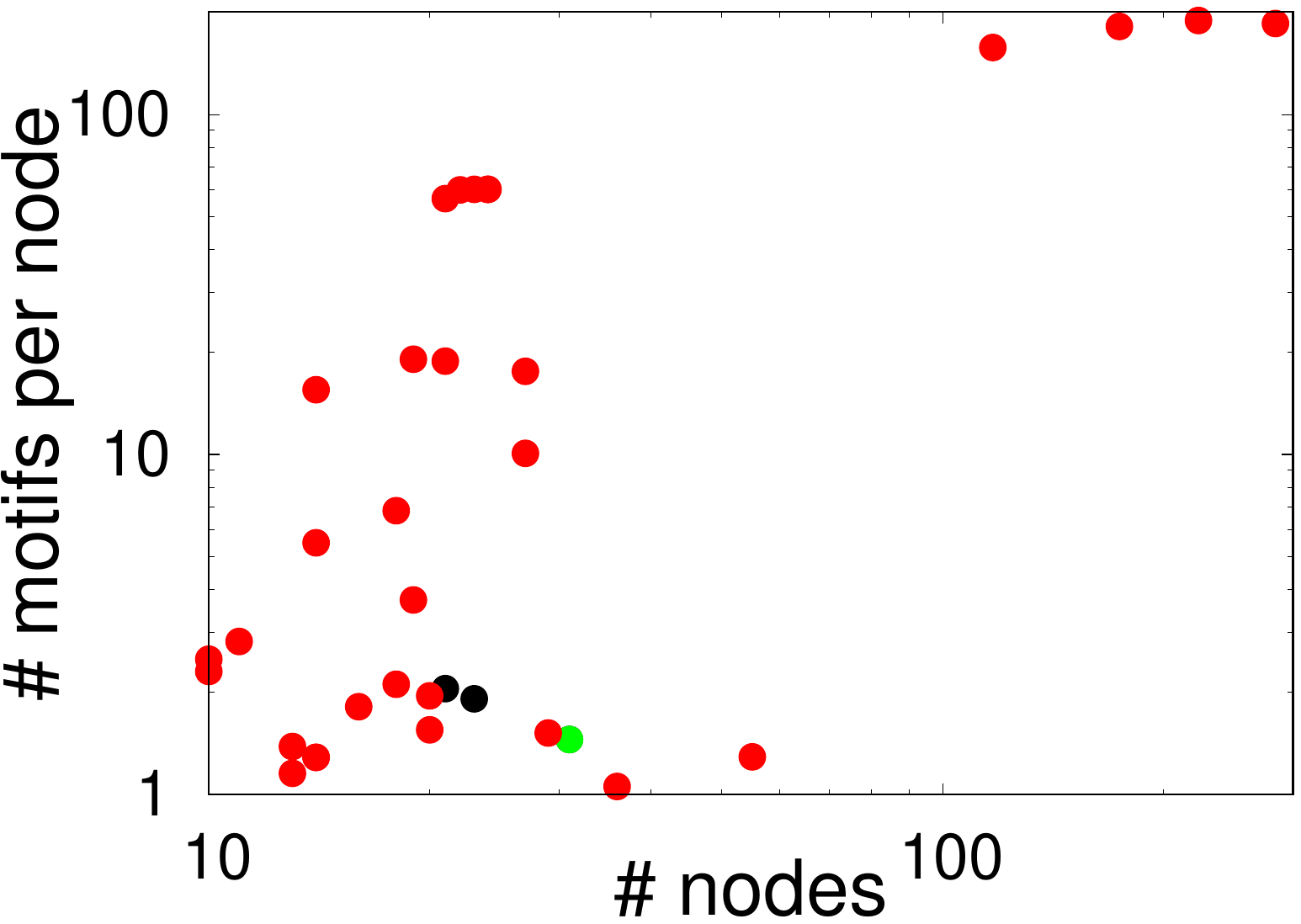}
	\vspace{-6ex}
        \caption{Average motif degree}
        \label{fig:bra}
\end{subfigure}
\hspace{-2ex}
\begin{subfigure}{.5\linewidth}
	\centering
	\includegraphics[width=\linewidth]{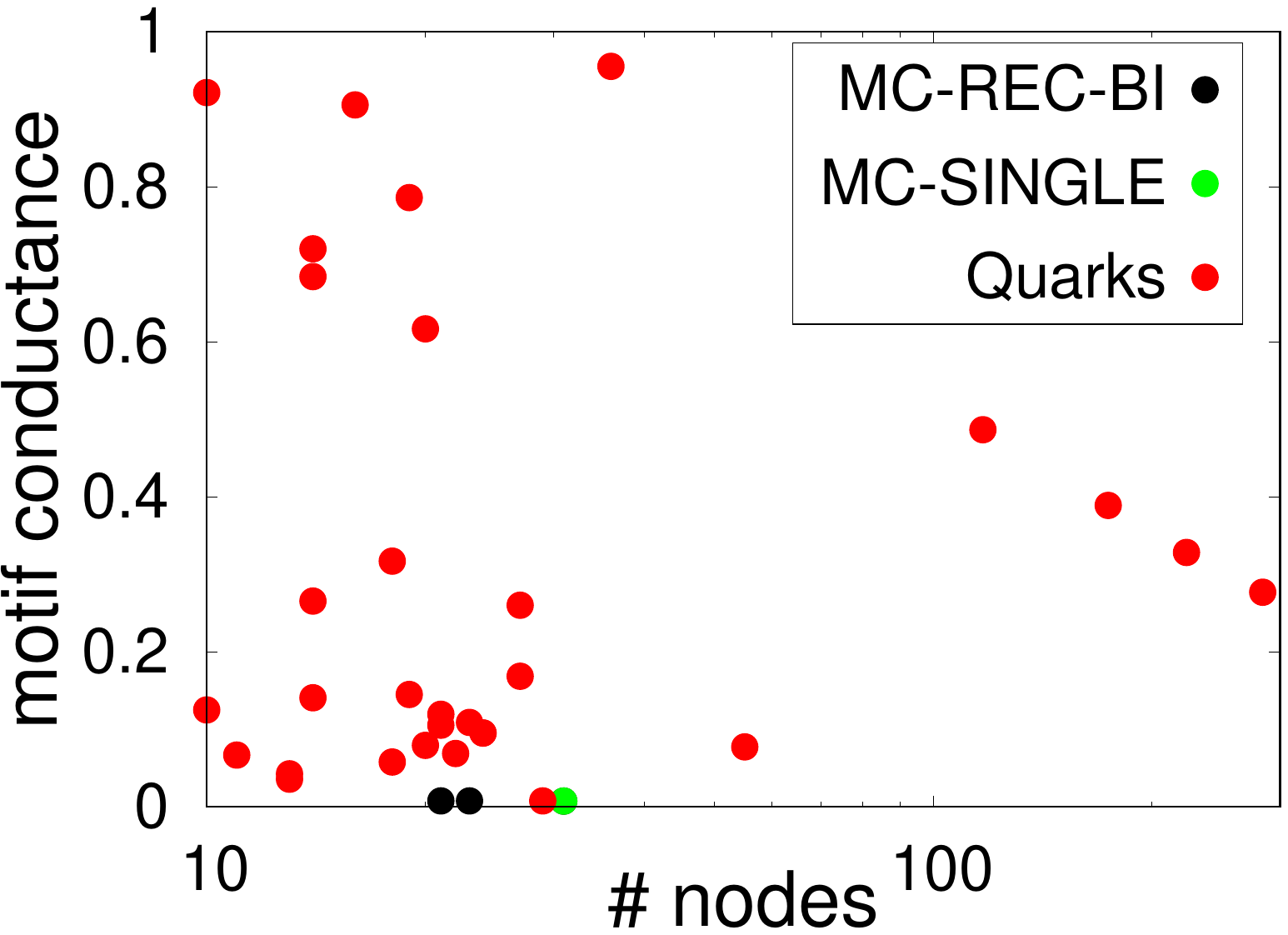}
	\vspace{-6ex}
        \caption{Motif conductance}
        \label{fig:brc}
\end{subfigure}
\vspace{-3ex}
\caption{\small Comparison of quarks and \MC~\cite{Benson16} for \ita{acyclic}{\tt +++} in \ti{reddit-body} network with respect to the number of motifs per node (left, higher is better) and motif conductance (right, lower is better). Subgraphs with at least 10 nodes are shown. Each subgraph is denoted by a point; size is on the x-axis and metric is on the y-axis.}
\label{fig:rba}

\end{figure}

\subsubsection{\bf Motif counts and quark numbers}

Acyclic variants are significantly more common than the cycles in all networks. 
Among the cycle variants, \texttt{+++} is the most prevalent in \ti{reddit} networks and \texttt{---} is the least common in all.
This is also coherent with the structural balance theory~\cite{Easley12}, which states that the triangles with an odd number of negative links are rare.
However, \ita{cycle}\texttt{++-} is more common than the other \ita{cycles} in the \ti{epinions} network.
This might be due to hierarchical status among the nodes; the lower status nodes are likely to trust the ones with higher status but the reverse is not true.
The ratio of balanced triangles is 0.8 for \ti{reddit} networks but 0.42 for \ti{epinions}.
Per the maximum quark numbers, we observe a correlation with the motif counts.
Among all, only \ita{acyclic}\texttt{+++} yields non-trivial subgraphs with large quark numbers.
\ita{Cycle} variants fail to give significant quarks.

\begin{table}
\centering
\small
\caption{\small Signed datasets. $|V|$, $|E|$, $|E_+|$, and $|E_-|$ are the number of nodes, edges, positive, and negative edges. Motif counts (see~\cref{fig:signedmotifs}) and corresponding maximum quark numbers are also shown.}
\vspace{-3ex}
\renewcommand{\tabcolsep}{2pt}
\begin{tabular}{|c|c|l|r|r|r|r|r|} \cline{1-7}
\multicolumn{3}{|c|}{} 	&	\ti{red-body}	&	\ti{red-title}	&	\ti{epinions}	&	\ti{slashdot}	\\	 	\cline{1-7}
\multicolumn{3}{|c|}{$|V|$}	&	 34.7K	&	 52.9K	&	 125.8K	&	 74.3K	 	\\	\cline{1-7}
\multicolumn{3}{|c|}{$|E|$}	&	 110.8K	&	 205.5K	&	 581.6K	&	 420.5K		\\	\cline{1-7}
\multicolumn{3}{|c|}{$|E_+|$}	&	 102.5K	&	 188.7K	&	 465.4K	&	 311.3K		\\	\cline{1-7}
\multicolumn{3}{|c|}{$|E_-|$}	&	 8.3K	&	 16.8K	&	 116.3K	&	 109.2K		\\	\cline{1-7}
&		&$	+++	$	&	 4.9K	&	 7.8K	&	 14.1K	&	 1.2K		\\	
&		&$	++-	$	&	 1.3K	&	 1.9K	&	 34.2K	&	 1.3K		\\	
&	\ita{cycle}	&$	+--	$	&	166	&	233	&	 10.9K	&	702		\\	
&		&$	---	$	&	9	&	13	&	745	&	86		\\	\cline{2-7}
&		&$	+++	$	&	 145.7K	&	 592.9K	&	 1.1M	&	 125.0K		\\	
motif &		&$	++-a	$	&	 20.0K	&	 88.8K	&	 37.1K	&	 15.3K		\\	
counts &		&$	++-b	$	&	 22.1K	&	 88.8K	&	 14.1K	&	 9.0K		\\	
&	\ita{acylic}	&$	++-c	$	&	 18.0K	&	 59.3K	&	 115.0K	&	 15.7K		\\	
&		&$	+--a	$	&	 3.4K	&	 10.8K	&	 88.5K	&	 11.8K		\\	
&		&$	+--b	$	&	 3.1K	&	 10.8K	&	 8.3K	&	 7.8K		\\	
&		&$	+--c	$	&	 6.2K	&	 26.6K	&	 92.1K	&	 30.4K		\\	
&		&$	---	$	&	 1.1K	&	 3.5K	&	 39.7K	&	 9.0K		\\	\hline \hline
max. & \ita{cycle}		&    all	   	&   	1	   &   	1-2	   &   	1-2	   &   	1	   	\\	\cline{2-7}
quark & \ita{acyclic}		&$	+++	$	&   	15	   &   	20	   &   	15	   &   	6	   	\\	
numbers &		&   	others	   	&   	2-3	   &   2-4	   &   	2-5	   &   	2-5	   	\\		\hline
\end{tabular}
\label{tab:signeds}
\end{table}

\subsubsection{\bf Comparison with \MC}
We compare the $k$-quarks with the \MCsingle and \MCrecbi for \ita{acyclic}{\tt +++} motif in \cref{fig:rba}.
Some quarks are able to obtain very low conductance scores, close to \MC results.
For the average motif degrees, quarks significantly outperform \MCrecbi.
In particular, one quark has 284 nodes with average motif degree of 185.2.

One of the quarks by \ita{acyclic}\texttt{+++} has 21 subreddits about rappers/singers, such as \textit{kanye} and \textit{kendricklamar}.
The ones which praised the others but have not received much praise (the white node in \ita{acyclic} in~\cref{fig:dirmotifs}) are \textit{boogalized, runthejewels,} and ~\textit{charlieputh}. The last two are a young rapper duo and a new Canadian singer, respectively, with 6.6K and 279 members.
On the other hand, the subreddits that got praised by the others but have not reciprocated (the black node in \ita{acyclic}) are \textit{theweeknd, frankocean,} and \textit{kidcudi}.
Those are experienced ones (active since 2010, 2005, and 2003) with tens of thousands of members in their subreddits.

\vspace{-1ex}
\subsection{Node-labeled networks}\label{sec:expnodelabeled}

{\bf Datasets.} Here we consider node-labeled undirected networks. We use the Facebok100 dataset that contains the complete Facebook networks of 100 American colleges from a single-day snapshot in September 2005~\cite{traud2012social}. Each node has multiple labels, here we only consider the genders of the nodes (there are only two available in the dataset; female and male) and use quark decomposition to find subgraphs that have balanced gender ratios, i.e., close number of females and males.
Excluding the female-only institutions, the overall average female ratio is $\%48.5$ and there are 57 networks with less than $50\%$ female.
We choose 18 networks with the lowest female ratio (all have $<45\%$),~\cref{tab:femaleratios} gives a partial list.

\begin{table}[!t]
\centering
\small
\caption{\small Node-labeled graphs. Statistics and avg. of the female ratios in nuclei and quarks are given. $V_f$ denotes the set of female nodes. (2,3)n and (3,4)n are the nucleus decompositions.}
\vspace{-3ex}
\renewcommand{\tabcolsep}{0.5pt}
\begin{tabular}{|l|r|r|c||c|c|c|c|c|c|c|} \hline
	&	\multicolumn{1}{c|}{\multirow{3}{*}{$|V|$}} 	&	\multicolumn{1}{c|}{\multirow{3}{*}{$|E|$}}	&	\multicolumn{1}{c||}{\multirow{3}{*}{$\frac{|V_f|}{|V|}$}} 	&	\multicolumn{3}{c|}{edge, triangle}	&	\multicolumn{4}{c|}{triangle, 4-clique}					\\						
	&		&		&		&	\multicolumn{1}{c|}{\multirow{2}{*}{(2,3)n}}	&	 \multicolumn{2}{c|}{Quarks} 	&	\multirow{2}{*}{(3,4)n}	&	\multicolumn{3}{c|}{Quarks}	\\	
	&		&		&		&		&	{\tt FMM}	&	{\tt FFM}	& 		&	{\tt FMMM}	&	{\tt FFMM}	&	{\tt FFFM}	\\ \hline	
{\ti Mich67}	&	 3.7K	&	 81.9K	&	25\%	&	23.0\%	&	45.0\%	&	50.0\%	&	24.5\%	&	40.0\%	&	45.0\%	&	51.6\%	\\
{\ti Caltech36}	&	769	&	 16.7K	&	30\%	&	39.4\%	&	46.0\%	&	52.0\%	&	38.5\%	&	43.1\%	&	50.2\%	&	52.8\%	\\
{\ti Carnegie49}	&	 6.6K	&	 250.0K	&	37\%	&	32.6\%	&	49.0\%	&	52.5\%	&	38.5\%	&	43.5\%	&	49.5\%	&	54.9\%	\\
{\ti MIT8}	&	 6.4K	&	 251.3K	&	37\%	&	38.8\%	&	48.0\%	&	52.1\%	&	42.0\%	&	44.3\%	&	50.3\%	&	53.9\%	\\
{\ti Stanford3}	&	 11.6K	&	 568.3K	&	40\%	&	46.8\%	&	48.1\%	&	49.0\%	&	44.1\%	&	45.4\%	&	49.2\%	&	55.4\%	\\
{\ti Cornell5}	&	 18.7K	&	 790.8K	&	44\%	&	44.3\%	&	47.6\%	&	51.8\%	&	45.6\%	&	43.7\%	&	48.7\%	&	54.9\%	\\
{\ti Penn94}	&	 41.6K	&	 1.4M	&	44\%	&	49.7\%	&	48.4\%	&	51.4\%	&	52.1\%	&	44.0\%	&	49.8\%	&	55.8\%	\\
{\ti UPenn7}	&	 14.9K	&	 686.5K	&	44\%	&	37.3\%	&	48.8\%	&	51.1\%	&	46.4\%	&	45.1\%	&	50.4\%	&	55.4\%	\\ \hline
\multicolumn{3}{|c|}{\bf Average of 18 networks:}	&					40\%	&	42.5\%	&	48.2\%	&	51.5\%	&	44.1\%	&	44.4\%	&	49.7\%	&	54.7\%	\\ \hline
\end{tabular}
\label{tab:femaleratios}
\end{table}

\noindent {\bf Motifs.} We instantiate the quark decomposition in five ways, where {\tt F/M} denotes the female/male nodes:
(1) $M$ is vanilla edge and $N$ is triangle in the following two forms: {\tt FMM} and {\tt FFM}; and
(2) $M$ is vanilla triangle, $N$ is four-clique in the following three forms: {\tt FMMM}, {\tt FFMM}, and {\tt FFFM}.
Also, there is no role confusion for any variant since the graph is undirected and node labels ensure that an edge cannot serve in different roles in its triangles in (1) (likewise for (2)).

\subsubsection{\bf Finding gender balanced subgraphs}

Algorithmic fairness is one of the most important problems in today's automated world~\cite{kearns2019ethical}.
Algorithms can amplify the implicit bias in the data, particularly based on the protected attributes like gender, race, ethnicity, and this can lead to unwanted consequences in criminal justice system, hiring, credit scoring, and more~\cite{barocas2017fairness}.
The bias in the network data is more complicated; regarding the gender attribute, for example, the problem is not only the imbalanced gender distribution but also how each gender category is connected to the other categories.
There are a few studies that analyze the implications on information diffusion~\cite{stoica2019fairness, Zeinab20}.
In this context, the community structure in the network plays an important role and algorithms that do not actively consider the protected attributes are likely to fail getting fair results.
Algorithms that can find balanced communities even when there is an imbalance in the input network are essential.
Here we use quark decomposition to find subgraphs with balanced gender ratios.
As explained above, we set the input motif $N$ in ways to reflect the characteristics of gender balanced subgraphs.

~\cref{tab:femaleratios} gives the female ratios in quarks in comparison to the label-oblivious nucleus decomposition algorithms~\cite{Sariyuce15}.
For each network, we find the leaf quarks (\cref{def:hier}) and nuclei with at least 10 nodes, calculate the female ratio in each quark or nucleus, and then take the average of those ratios.
We also show the average ratios across all the 18 datasets at the bottom.
The average female ratio of all the networks is $40.3\%$.
$(2,3)$-nuclei have a bit better number, $42.5\%$.
Quarks for {\tt FMM} get $48.2\%$ and the ratios are consistently good for all the networks, varying between $45\%-51.3\%$.
This significant jump from $(2,3)$-nuclei is due to the fact that each edge {\it has to} participate in a number of {\tt FMM} triangles.
Quarks for {\tt FFM} give an even better ratio, $51.5\%$.
$(3,4)$-nuclei are larger in number and also denser than $(2,3)$.
Its average female ratio ($44.1\%$) is also a bit better than $(2,3)$-nuclei.
Quarks for {\tt FMMM} are very close to $(3,4)$-nuclei but more consistent. 
{\tt FFMM} achieves $49.7\%$ and {\tt FFFM} gets the best: $54.7\%$.
Note that results get better for the motifs with larger female ratio: {\tt FMMM} < {\tt FMM} < {\tt FFMM} < {\tt FFM} < {\tt FFFM}.

Quarks cannot provide a good theoretical lower bound for the female ratio since there is no size constraint in the quark definition. For instance, a 1-quark for {\tt FFM} can possibly be formed by a pair of connected female nodes and $n$ male nodes that are connected to both females; the female ratio would be $2/(n+2)$ in this case.
But in practice, quarks with female-dominant motifs yield dense subgraphs with high female ratios. Even with {\tt FMM}, which implies a 1/3 ratio (smaller than the average female ratio of the datasets), quarks can obtain better results than the label-oblivious (2,3)-nuclei.

\begin{figure}[!t]
\centering
\hspace{-6ex}
\begin{subfigure}{\linewidth}
	\centering
	\hfill \includegraphics[width=0.35\linewidth]{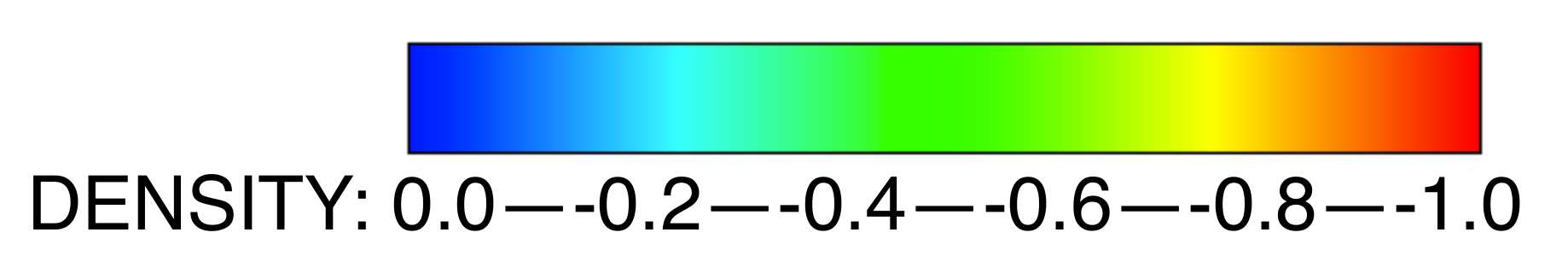}
	\vspace{-10ex}
\end{subfigure}
\begin{subfigure}{.5\linewidth}
	\centering
	\includegraphics[width=\linewidth]{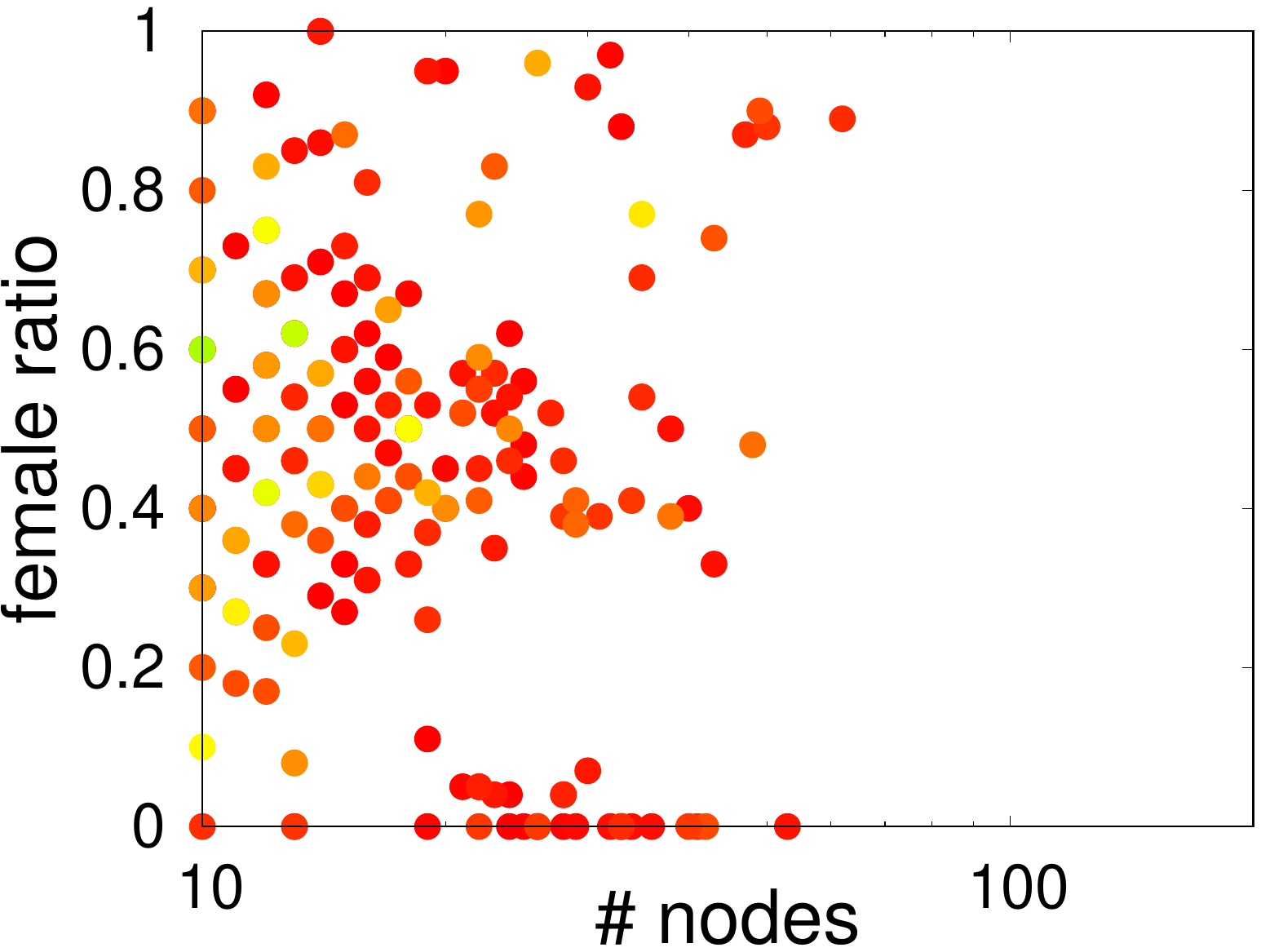}
	\vspace{-5ex}
        \caption{$(3,4)$ nuclei}
        \label{fig:upenn34}
\end{subfigure}
\hspace{-3ex}
\begin{subfigure}{.5\linewidth}
	\centering
	\includegraphics[width=\linewidth]{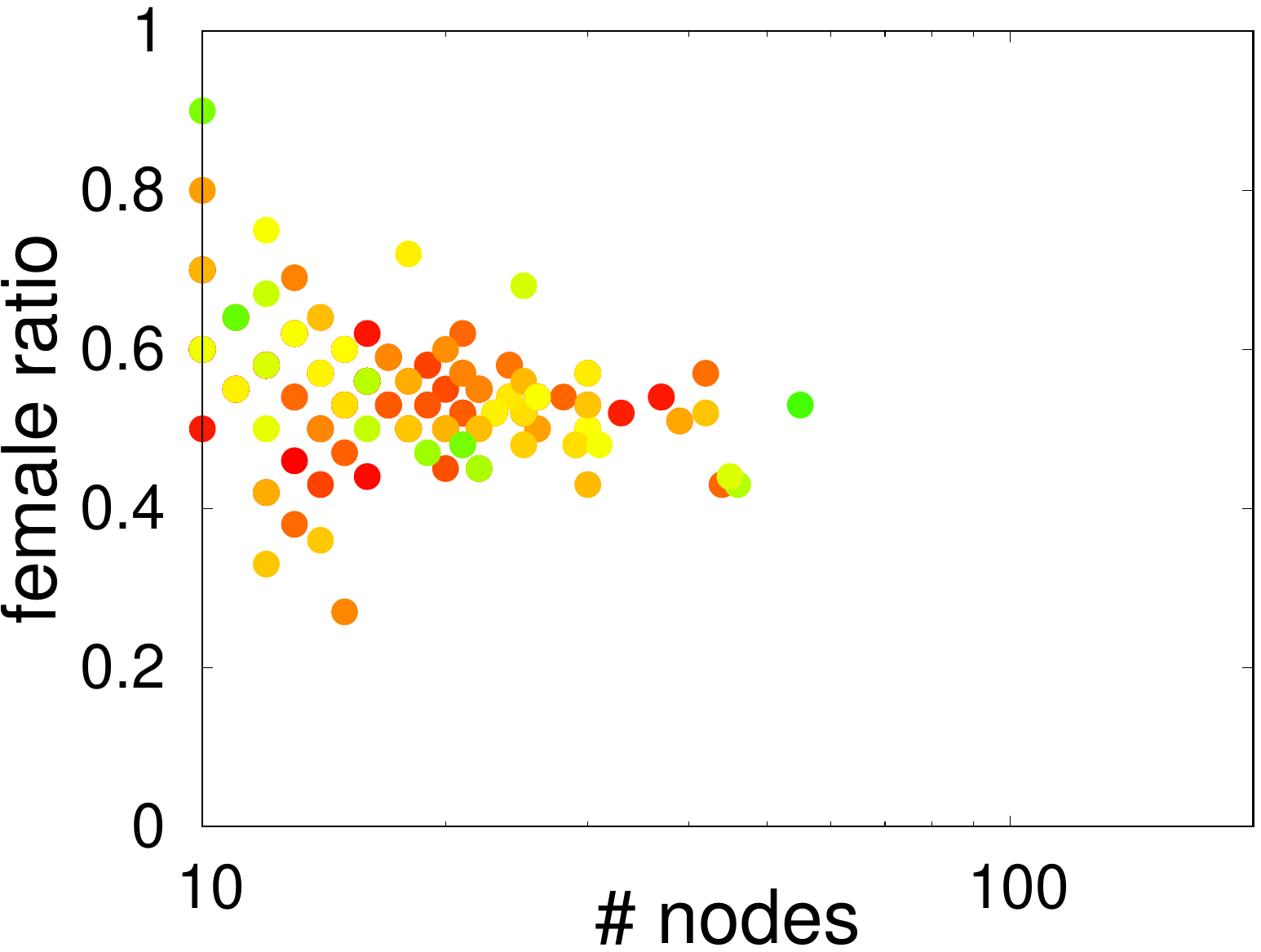}
	\vspace{-5ex}
        \caption{Quarks by {\tt FFFM}}
        \label{fig:upennFFFM}
\end{subfigure}
\vspace{-3ex}
\caption{\small Comparison of the female ratios in subgraphs obtained by quark ({\tt FFFM}) and $(3,4)$-nucleus decompositions for {\ti UPenn7}. For each subgraph, size is given on the x-axis, female ratio is shown on the y-axis, and the edge density is color-coded.}
\label{fig:upenn}
\end{figure}

We also show all the subgraphs with size, edge density, and female ratio information obtained by quark and nucleus decompositions in {\ti UPenn7} network.
~\cref{fig:upenn} gives $(3,4)$-nuclei and quarks for {\tt FFFM} on {\ti UPenn7} network.
Quarks are consistently gender balanced when compared to the nuclei; no quark with less than $25\%$ female ratio exists.
Note that there is a bit degradation in the number and density of the quarks for {\tt FFFM}: 216 subgraphs with 0.88 avg. edge density, compared to the 230 $(3,4)$ nuclei with avg. density 0.94.
Given the consistently high female ratios, we believe that this is an affordable loss in quality. 

\section{Discussion}

Quark decomposition offers a principled approach for motif-driven dense subgraph discovery in heterogeneous networks by successfully regularizing the motif degrees to quark numbers.
Our evaluation shows that the $k$-quarks can find dense subgraphs according to a given motif.
Role-aware variant solves the role confusion problem by creating multiple quark numbers for each motif $M$.
Overall, quark decomposition is versatile, efficient, and extendible.

For future work, it would be interesting to investigate the other byproducts of the quark decomposition, such as hierarchy structure. Our initial results show limited success; detailed and meaningful hierarchies are rare for the most motifs. Theoretical and empirical analysis of the impact of the input motifs, $M, N$, on the hierarchy structure would be interesting .
Also, adapting the quark decomposition for numerical attributes on nodes/edges would be promising.

\vspace{-1ex}
\begin{acks}
This research was supported by NSF-1910063 award, JP Morgan Chase and Company Faculty Research Award, and used resources of the Center for Computational Research at the University at Buffalo~\cite{CCR} and the National Energy Research Scientific Computing Center, a DOE Office of Science User Facility supported by the Office of Science of the U.S. Department of Energy under Contract No. DE-AC02-05CH11231.
\end{acks}

\bibliographystyle{ACM-Reference-Format}
\bibliography{paper}

\clearpage
\newpage

\end{document}